\documentclass{article}
\usepackage[utf8]{inputenc}
\usepackage{amsfonts}
\usepackage{mathtools}
\usepackage{amsmath}
\usepackage{amssymb}
\usepackage{amsthm}
\usepackage{natbib}
\usepackage{babel}
\usepackage{graphicx}
\usepackage{graphics}
\usepackage[font=small,labelfont=bf]{caption}
\usepackage{xcolor}
\usepackage{bbm}
\usepackage{hyperref}
\usepackage{pgfplots}
\usepackage{amsthm}
\usepackage{booktabs}
\usepackage{multirow}
\pgfplotsset{width = 5cm, compat = 1.17}
\usepackage[margin=1.2in]{geometry}

\title{Testing distributional equality for functional random variables}
\author{Bilol Banerjee}
\date{Theoretical Statistics and Mathematics Unit\\
      Indian Statistical Institute, Kolkata\\%
      Email : banerjeebilol@outlook.com\\
   \null }

\newtheorem{prop}{Proposition}[section]
\newtheorem{thm}{Theorem}[section]

\newtheorem{cor}{Corollary}
\newtheorem{rem}{Remark}
\newtheorem{lemmaA}{Lemma A.\ignorespaces}

\renewcommand{\P}{\mathbb{P}}

\newcommand{\R}{\mathbb{R}}
\newcommand{\E}{\mathbb{E}}
\renewcommand{\H}{\mathcal{H}}

\newcommand{\tikzcircle}[2][red,fill=red]{\tikz[baseline=-0.5ex]\draw[#1,radius=#2] (0,0) circle ;}

\definecolor{applegreen}{rgb}{0.55,0.71,0.0}
\begin{document}

\maketitle
\begin{abstract}

In this article, we present a nonparametric method for the general two-sample problem involving functional random variables modelled as elements of a separable Hilbert space ${\cal H}$. First, we present a general recipe based on linear projections to construct a measure of dissimilarity between two probability distributions on ${\cal H}$. In particular, we consider a measure based on the energy statistic and present some of its nice theoretical properties. A plug-in estimator of this measure is used as the test statistic to construct a general two-sample test. Large sample distribution of this statistic is derived both under null and alternative hypotheses. However, since the quantiles of the limiting null distribution are analytically intractable, the test is calibrated using the permutation method. We prove the large sample consistency of the resulting permutation test under fairly general assumptions. We also study the efficiency of the proposed test by establishing a new local asymptotic normality result for functional random variables. Using that result, we derive the asymptotic distribution of the permuted test statistic and the asymptotic power of the permutation test under local contiguous alternatives. This establishes that the permutation test is statistically efficient in the Pitman sense. Extensive simulation studies are carried out and a real data set is analyzed to compare the performance of our proposed test with some state-of-the-art methods.


\end{abstract}

\section{Introduction}

In a two-sample problem, we test for the equality of two distributions $F$ and $G$ based on two sets of independent observations $\mathcal{X} = \{X_1,X_2,\ldots,X_n\}$ and $\mathcal{Y} = \{Y_1,Y_2\ldots,Y_m\}$ on $X\sim F$ and $Y\sim G$, respectively.
For multivariate data, several two-sample tests are available in the literature. Notable methods include the tests based on average inter-point distances
\citep{szekely2004testing,baringhaus2004new,baringhaus2010rigid,biswas2014nonparametric}, graph based tests \citep{friedman1979multivariate,rosenbaum2005exact,biswas2014distribution}, test based on nearest neighbour type coincidences \citep{schilling1986multivariate,henze1988multivariate,mondal2015high} and those based on kernels \citep{gretton2012kernel,gretton2009fast}. However, in the case of functional data, the literature on the general two-sample test is scarce.

In functional data analysis, the random variables are often modelled as elements of an Hilbert space such as $L_2([a,b])$ or $L_2(\mathcal{D})$, the space of all square-integrable functions defined on the domain $\mathcal{D}\subset \R^p$ ($p\geq 1$) endowed with the $L_2$ metric \citep[see][]{fdabookramsay,ferraty2006nonparametric,hsing2015theoretical}. For such random variables, there are many ANOVA-type tests \citep{zhang2010two,cuesta2010simple,qiu2021two} that deal with the location problem.  \cite{hall2007two} proposed a Cramer-von-Mises type test for the general two-sample problem involving functional data. \cite{pomann2016two} suggested applying the Anderson-Darling test on the first few functional principal components of the mixture distribution and aggregating the results using Bonferroni's correction. \cite{wynne2020kernel} developed a test based on kernel mean embedding of the distributions of functional random variables. \cite{pan2018ball} proposed a test based on ball divergence between the distributions of two Banach-valued random variables. Most of these tests are based on a consistent estimate of a measure of dissimilarity between the two underlying distributions, and they have large sample consistency. However, the exact or limiting null distributions of these test statistics are usually analytically intractable and the permutation method is used for calibration. But the existing literature is somewhat silent about the statistical efficiency of these tests, partly because of the difficulty in formulating a suitable notion of the density and the likelihood ratio statistic.

In this article, we assume $X\sim F$ and $Y\sim G$ to be independent functional random variables lying in an infinite dimensional separable Hilbert space ${\mathcal H}$ with inner product $\langle.,.\rangle$. We know that two $\mathcal{H}$-valued random variables $X$ and $Y$ have the same distribution (i.e., $F=G$) if and only if the random variables $\langle X,f\rangle$ and $\langle Y,f\rangle$ are identically distributed for all $f\in\mathcal{H}$. So, the information about the dissimilarity between $F$ and $G$ is supposed to be contained in the distributions of the linear projections $\langle X,f\rangle$ and $\langle Y,f\rangle$ (denoted by $F^f$ and $G^f$, respectively) for $f\in\H$. One can use a suitable measure of dissimilarity $T(F^f, G^f)$ between two univariate distributions $F^f$ and $G^f$ and aggregate them over $f \in {\cal H}$ to come up with a general measure of dissimilarity between $F$ and $G$. In Section 2, we discuss this recipe for constructing a measure of dissimilarity between $F$ and $G$ and study some of its theoretical properties. For a suitable choice of $T(\cdot,\cdot)$, the proposed measure turns out to be non-negative, and under very general assumptions, it takes the value zero if and only if $F=G$. In this article, we use the measure based on energy statistics proposed in \cite{baringhaus2010rigid} as $T$, and it ensures this characterization property. In Section 3, we propose a consistent estimator of our measure and use it as the test statistic to test for the equality of $F$ and $G$. Large sample distribution of the test statistic is derived both under fixed null and alternative hypotheses. These results establish the consistency of our test even when the sample sizes are extremely unbalanced. However, the limiting null distribution of the test statistic is analytically intractable. So, we use the conditional test based on the permutation principle and prove its large sample consistency. We also establish a local asymptotic normality result for functional random variables. To the best of our knowledge, such results are new in the functional data analysis literature. This result helps us to construct a locally asymptotically normal sequence of contiguous alternatives and study the behaviour of our test under such alternatives. Our results show that the proposed test is statistically efficient in the Pitman sense, i.e., under such contiguous alternatives, the power of the permutation test converges to a non-trivial limit as the sample sizes increase. 
Extensive simulation studies are carried out and a real data set is analyzed in Section 4 to compare our performance with some state-of-the-art methods. Finally, Section 5 contains some concluding remarks and a brief discussion on possible future directions. All proofs and mathematical details are deferred to the Appendix.

\section{Measure of dissimilarity for functional random variables}

Let $\mathcal{H}$ be a separable Hilbert space with inner product $\langle.,.\rangle$ and  $\mathcal{B}(\mathcal{H})$ be the Borel $\sigma$-field on $\mathcal{H}$. Consider a random variable $Z$  that takes values on $\mathcal{H}$. We know that (i) $Z$ is $\mathcal{B}(\mathcal{H})$-measurable if and only if $\langle Z,f\rangle$ is measurable for all $f\in\mathcal{H}$ and (ii) the distribution of $Z$ is uniquely determined by the distributions of $\langle Z,f\rangle$ over $f\in\mathcal{H}$
\citep[see, e.g., Theorem 7.1.2 in][]{hsing2015theoretical}.
So, two $\mathcal{H}$-valued random variables $X$ and $Y$ have the same distribution if and only if the random variables $\langle X,f\rangle$ and $\langle Y,f\rangle$ are identically distributed for all $f\in\mathcal{H}$. Now, consider any measure of difference $T(\cdot, \cdot)$ between two univariate distributions, which is non-negative and takes the value zero if and only if the two distributions are equal. One can use it to measure the difference between $F^{f}$ and $G^{f}$, the distributions corresponding to $\langle X,f\rangle$ and $\langle Y,f\rangle$, and aggregate them over $f \in \mathcal{H}$ to come up with a measure of dissimilarity between $F$ and $G$. This can be expressed as
$$\zeta^\nu(F,G) = \int_{\H} T(F^f,G^f)\,d\nu(f),$$
where $\nu$ is some probability measure on $\mathcal{H}$. It is easy to see that if $F$ and $G$ are identical, then $\zeta^\nu(F,G) = 0$. But $\zeta^\nu(F,G) = 0$ only implies that $F^f = G^f$ almost everywhere w.r.t. $\nu$, which does not necessarily imply $F=G$. Note that in the multivariate case, if $T$ is chosen as the squared $L_2$-distance between $F^f$ and $G^f$ and $\nu$ is chosen as the uniform distribution over the surface of the unit sphere in $\R^d$, $\zeta^\nu(F,G)$ turns out to be the energy distance between $F$ and $G$ \citep{baringhaus2004new}, and in that case, $\zeta^\nu(F,G)$ has the characterization property, i.e. $\zeta^\nu(F,G)=0$
implies $F=G$. If $T$ is the Cramer-von-Mises distance between $F^f$ and $G^f$, the same choice of $\nu$ leads to the two-sample test statistic proposed in \cite{kim2020robust}. It also has the characterization property. In these two cases, $\nu$ being the uniform distribution has support over the entire surface of the unit ball and hence considers all possible directions for projection. Keeping that in mind, we can consider a probability measure $\nu$, whose support contains the unit sphere centered at the origin of the Hilbert space. In that case, $\zeta^\nu(\cdot,\cdot)$ has the characterization property, as shown in the following theorem.

\begin{thm}
If $supp\{\nu\}$ contains the unit sphere in $\H$, then $\zeta^\nu(F,G) = 0$ if and only if $F=G$.
\label{thm1}
\end{thm}

However, note that the $f$'s, which are orthogonal to $supp\{F\}\cup supp\{G\}$, do not contribute to $\zeta^{\nu}(F,G)$ even when the two random variables $X$ and $Y$ are highly separated. Therefore, it seems reasonable to discard those directions and work with $\nu= (F+G)/2$, an equal mixture of $F$ and $G$. It turns out that the characterization property of $\zeta^{\nu}(F,G)$ holds for this choice of $\nu$ as well. This is formally stated in the following theorem.

\begin{thm}
If $\nu = (F+G)/2$, then, $\zeta^{\nu}(F,G) = 0$ if and only if $F=G$.
\label{thm2}
\end{thm}

Throughout this article, we use $\nu=(F+G)/2$ while $T$ is taken as the measure proposed in \cite{baringhaus2010rigid}, which is defined as

$$T_\phi(\mathcal{L}_1,\mathcal{L}_2) = 2\E_{\mathcal{L}_1,\mathcal{L}_2}~\phi\big(|U-V|^2\big)-\E_{\mathcal{L}_1,\mathcal{L}_1}~\phi\big(|U-U^\prime|^2\big) - \E_{\mathcal{L}_2,\mathcal{L}_2}~\phi\big(|V-V^\prime|^2\big),$$

where $U,U^{'} \stackrel{iid}{\sim} {\cal L}_1$ and $V,V^{'} \stackrel{iid}{\sim} {\cal L}_2$ are independent random variables, $\phi:[0,\infty)\to[0,\infty)$ is continuous, monotinically increasing function with $\phi(0) = 0$, and it has non-constant completely monotone derivative on $(0,\infty)$ with $\E_{\mathcal{L}_1}~\phi(|U|^2)$ and $\E_{\mathcal{L}_2}~\phi(|V|^2)$ being finite. For this choice of $T$, the measure of dissimilarity between $F$ and $G$ is given by
$$\zeta_\phi(F,G):= \frac{1}{2}\int_{\H} T_\phi(F^f,G^f)dF(f)+\frac{1}{2}\int_{\H} T_\phi(F^f,G^f)dG(f).$$
Since $\zeta_\phi(F,G)$ is obtained by aggregating the Baringhaus-Franz statistic $T_\phi$ computed along different projection directions, we call it the projected BF (pBF) criterion. It has a closed form expression given by
\vspace{-0.05in}
\begin{equation*}
    \begin{split}
        \zeta_\phi(F,G) 
        = & ~ \E_{X_3}\Big[\E_{{X_1},{Y_1}}~\phi\big(|\langle X_1,X_3\rangle-\langle Y_1,X_3\rangle|^2\big)\Big] -\frac{1}{2}\E_{X_3}\Big[\E_{X_1,X_2}~\phi\big(|\langle X_1,X_3\rangle-\langle X_2,X_3\rangle|^2\big)\Big]\\
        & - \frac{1}{2}\E_{X_3}\Big[\E_{Y_1,Y_2}~\phi\big(|\langle Y_1,X_3\rangle-\langle Y_2,X_3\rangle|^2\big)\Big] + \E_{Y_3}\Big[\E_{{X_1},{Y_1}}~\phi\big(|\langle X_1,Y_3\rangle-\langle Y_1,Y_3\rangle|^2\big)\Big]\\
        & - \frac{1}{2}\E_{Y_3}\Big[\E_{X_1,X_2}~\phi\big(|\langle X_1,Y_3\rangle-\langle X_2,Y_3\rangle|^2\big)\Big] - \frac{1}{2}\E_{Y_3}\Big[\E_{Y_1,Y_2}~\phi\big(|\langle Y_1,Y_3\rangle-\langle Y_2,Y_3\rangle|^2\big)\Big],
    \end{split}
\end{equation*}
where $X_i\stackrel{i.i.d.}{\sim}F\,(i=1,2,3)$, $Y_i\stackrel{i.i.d.}{\sim}G\,(i=1,2,3)$ are independent and $\E~\phi(|\langle X_1,X_2 \rangle|^2)$, $\E~\phi(|\langle X_1,Y_1 \rangle|^2)$ and $\E~\phi(|\langle Y_1,Y_2 \rangle|^2)$ are finite. The measure $\zeta_\phi(F,G)$ has some nice theoretical properties as mentioned in the following proposition. 

\begin{prop}
Suppose that $\phi:[0,\infty)\to[0,\infty)$ is continuous, monotinically increasing function with $\phi(0) = 0$, and it has non-constant completely monotone derivative on $(0,\infty)$. Also assume that $\E_{F^f}~\phi(|U|^2)$ and $\E_{G^f}~\phi(|V|^2)$ are finite for all $f \in{\cal H}$. Then $\zeta_\phi(F,G)$ has the following properties
\begin{itemize}
    \item[(a)] $\zeta_\phi(F,G) = \E\{g(X_1,X_2,X_3; Y_1,Y_2,Y_3)\}$, where
    \begin{equation*}
        \begin{split}
            g(X_1,X_2,X_3; Y_1,Y_2,Y_3)=\frac{1}{2}&\Big\{2\phi\big(|\langle X_1,X_3\rangle-\langle Y_1,X_3\rangle|^2\big) -\phi\big(|\langle X_1,X_3\rangle-\langle X_2,X_3\rangle|^2\big)\\
            & - \phi\big(|\langle Y_1,X_3\rangle-\langle Y_2,X_3\rangle|^2\big) + 2\phi\big(|\langle X_1,Y_3\rangle-\langle Y_1,Y_3\rangle|^2\big)\\
            & - \phi\big(|\langle X_1,Y_3\rangle-\langle X_2,Y_3\rangle|^2\big) -  \phi\big(|\langle Y_1,Y_3\rangle-\langle Y_2,Y_3\rangle|^2\big)\Big\}.
        \end{split}
    \end{equation*}
    \item[(b)] $\zeta_\phi(F,G)$ has the distribution characterization property, i.e., $\zeta(F,G) = 0$ if and only if $F = G$.
    
    \item[(c)] $\zeta_\phi(F,G)$ is invariant under unitary operations on $X$ and $Y$, i.e., if $U:\H\to \mathcal{X}$ is an unitary operator, then $\zeta(F\circ U^{-1},G\circ U^{-1}) = \zeta(F,G)$.
    
    \item[(d)] If $\{{X}_n:n\geq 1\}$ and $\{{Y}_n:n\geq 1\}$ are independent sequences of Hilbertian random variables such that ${X}_n\Rightarrow {X}$ and ${Y}_n\Rightarrow {Y}$, then $\lim\limits_{n\to\infty}\zeta_\phi(\mathcal{L}({X}_n),\mathcal{L}({Y}_n))=\zeta_\phi(\mathcal{L}({X}),\mathcal{L}({Y})).$
\end{itemize}
\label{depprop}
\end{prop}

\begin{rem}
 Proposition \ref{depprop}(c) implies that $\zeta_\phi$ only depends on the inner product defined on the Hilbert space, but not on the space used for modeling the random variables. For example, modeling the two samples as random variables in $L_2[0,1]$ and in $L_2[0,10]$ leads to the same value of $\zeta_\phi(F,G)$. 
\end{rem}

\section{Estimation of pBF and construction of the two-sample test}
Suppose $\hat F_n$ and $\hat G_m$ be the empirical probability distribution functions based on the random samples $\mathcal{X}$ and $\mathcal{Y}$, respectively. Replacing $F$ by $\hat{F}_n$ and $G$ by $\hat{G}_m$, we get an estimator of $\zeta_\phi(F,G)$. This estimator $\hat \zeta_{n,m}^\phi:=\zeta_\phi(\hat{F}_n,\hat{G}_m)$ can be expressed as
\begin{equation*}
    \begin{split}
        \hat \zeta_{n,m}^\phi & = \frac{1}{n^2m}\sum_{i=1}^{n}\sum_{j = 1}^n\sum_{k=1}^m \phi\big(|\langle X_j,X_i\rangle-\langle Y_k,X_i\rangle|^2\big) - \frac{1}{2n^3}\sum_{i=1}^{n}\sum_{j=1}^n\sum_{k=1}^n \phi\big(|\langle X_j,X_i\rangle-\langle X_k,X_i\rangle|^2\big)\\
        & -\frac{1}{2nm^{2}}\sum_{i=1}^{n}\sum_{1\leq j, k\leq m} \phi\big(|\langle Y_j,X_i\rangle-\langle Y_k,X_i\rangle|^2\big)+\frac{1}{nm^2}\sum_{i=1}^{m}\sum_{j = 1}^n\sum_{k=1}^m \phi\big(|\langle X_j,Y_i\rangle-\langle Y_k,Y_i\rangle|^2\big)\\
        &- \frac{1}{2n^2m}\sum_{i=1}^{m}\sum_{j=1}^n\sum_{k=1}^n \phi\big(|\langle X_j,Y_i\rangle-\langle X_k,Y_i\rangle|^2\big)-\frac{1}{2m^{3}}\sum_{i=1}^{m}\sum_{1\leq j, k\leq m} \phi\big(|\langle Y_j,Y_i\rangle-\langle Y_k,Y_i\rangle|^2\big).
    \end{split}
\end{equation*}
Clearly, $\hat\zeta_{n,m}^\phi$ can be viewed as a two-sample V-statistic with the core function,
\begin{equation*}
            g^*(X_1,X_2,X_3; Y_1,Y_2,Y_3) \\
            = \frac{1}{3!3!}\sum_{\pi_1,\pi_2\in\mathcal{S}_3} g(X_{\pi_1(1)},X_{\pi_1(2)},X_{\pi_1(3)};Y_{\pi_2(1)},Y_{\pi_2(2)},Y_{\pi_2(3)}).
\end{equation*}
The raw computational complexity of this statistic is of the order $O(n^3+m^3)$, but one can reduce the cost by vectorization method in R. The large sample distribution of $\zeta_{n,m}^\phi$ is given by the following theorem.
\begin{thm}
Let $X_1,\ldots,X_n\stackrel{i.i.d.}{\sim} F$ and $Y_1,\ldots,Y_m\stackrel{i.i.d.}{\sim} G$ be independent random functions and $\lim n/(n+m)=\lambda\in[0,1]$. 
Then for any $\phi$ satisfying the properties mentioned in Proposition 2.1, as $\min\{n,m\}\to\infty$, we have the following results.
\begin{itemize}
    \item[(a)] Under $H_1:F\not=G$, $\sqrt{nm/(n+m)}(\hat\zeta_{n,m}^\phi-\zeta_\phi(F,G))$ converges in distribution to a normal random variable with mean zero and variance $(1-\lambda)\delta_1+\lambda\delta_2$ for some positive $\delta_1,\delta_2>0$.

    \item[(b)] Under $H_0:F=G$, $nm/(n+m)\hat\zeta_{n,m}^\phi$ converges in distribution to $\sum_{k=1}^\infty \lambda_k Z_k^2$ for some square integrable sequence $\{\lambda_k\}$ and independent standard normal sequence of random variables $\{Z_k\}$.
\end{itemize}
\label{largesampleres}
\end{thm}

As a consequence of Theorem \ref{largesampleres} we get the probability convergence of $\hat\zeta_{n,m}^\phi$ to its population counterpart $\zeta_\phi(F,G)$. This is formally stated as a corollary.
\begin{cor}
If ${X}_1,{X}_2,\ldots,{X}_n\stackrel{i.i.d.}{\sim}{F}$ and ${Y}_1,{Y}_2,\ldots,{Y}_m\stackrel{i.i.d.}{\sim}{G}$ are independent
$\hat\zeta_{n,m}^\phi$ converges in probability to $\zeta_\phi(F,G)$ as $\min\{n,m\}\to\infty$ (even when $n/(n+m)\to 0\text{ or }1$).
\label{consistency}
\end{cor}

Hence even in the extremely unbalanced scenario (i.e. when $n/(n+m)\to 0\text{ or }1$) our estimator can detect the distributional difference between the two samples $\mathcal{X}$ and $\mathcal{Y}$ of random functions.

\subsection{Two-sample test based on ${\hat \zeta}_{n,m}^\phi$ }

We have seen that for a suitable choice of $\phi$, we have $\zeta_\phi(F,G)\ge 0$, where the equality holds if and only if $F$ and $G$ are equal. Since $\hat\zeta_{n,m}^\phi$ is a consistent estimator of $\zeta(F,G)$, we can reject $H_0:F=G$ if $\hat\zeta_{n,m}^\phi$ is large. 
Theorem \ref{largesampleres} (b) gives us the limiting null distribution of the test statistics ${\hat \zeta}_{n,m}^{\phi}$, but it involves some unknown quantities which are quite difficult to estimate. Hence, for a level $\alpha$ ($0<\alpha<1$), the cut-off is computed using the permutation method as described below.

\begin{itemize}
    \item Let $\mathcal{U}^\pi = \{U_{\pi(1)},U_{\pi(2)},\ldots, U_{\pi(N)}\}$ denote a permutation of the pooled sample $\mathcal{U}$ based on the permutation $\pi$ of $\{1,2,\ldots, N\}$.
    
    \item Partition $\mathcal{U}^\pi$ into $\mathcal{X}_{n}^\pi = \{U_{\pi(1)},U_{\pi(2)},\ldots, U_{\pi(n)}\}$ and $\mathcal{Y}_{m}^\pi = \{U_{\pi(n+1)},U_{\pi(n+2)},\ldots, U_{\pi(n+m)}\}$ and compute the statistic $\hat\zeta_{n,m}^{\phi,\pi}$ (permutation analog of ${\hat \zeta}_{n,m}^\phi$).
    
    \item Return the critical value $c_{\alpha}^\phi$ defined by,
    $$c_{\alpha}^\phi = \inf\{t\in\R: \frac{1}{N!}\sum_{\pi\in \mathcal{S}_N}\mathbbm{1}[\hat\zeta_{n,m}^{\phi,\pi}\leq t]\geq 1-\alpha\},$$
    where $\mathcal{S}_N$ is the set of all permutations of $\{1,2,\ldots,N\}$.
\end{itemize}


 The proposed test rejects $H_0$ if $\hat\zeta_{n,m}^\phi$ is larger than $c_{\alpha}^\phi$ or equivalently the corresponding $p$-value $p_{n,m} = \frac{1}{N!}\big\{\sum_{\pi\in \mathcal{S}_N}\mathbbm{1}[\hat\zeta_{n,m}^{\phi,\pi}\geq \hat\zeta_{n,m}^\phi]\big\}$ is smaller than $\alpha$. Here the cut-off $c_{\alpha}^\phi$ is a random quantity, but using the following theorem one can prove that it converges to zero as $\min\{n,m\}$ diverges to infinity. 
\begin{thm}
 If $\phi$ satisfies the conditions mentioned in Proposition 2.1, 
 as $\min\{n,m\}$ grows to infinity, $nm/(n+m)\hat\zeta_{n,m}^{\phi,\pi}$ converges in distribution to $\sum_{k=1}^\infty \lambda_k Z_k^2$ where $\{Z_k\}$ and $\{\lambda_k\}$ are as in Theorem \ref{largesampleres}.
 \label{local-limit-per}
\end{thm}

In particular, under $H_0$, the permuted test statistic $nm/(n+m)\hat\zeta_{n,m}^{\phi,\pi}$ and the estimator $nm/(n+m)\hat\zeta_{n,m}^{\phi}$ attains the same limiting distribution as $\min\{n,m\}$ diverges to infinity. Hence the permutation test asymptotically attains the level of significance $\alpha$ and it turns out to be consistent for any fixed alternative. 




\begin{rem}
The proposed permutation test is consistent even when $n/(n+m)\to0\text{ or }1$. Hence, even for the extremely unbalanced scenario, it performs well if $\min\{n,m\}$ is sufficiently large. 
\end{rem}


However, in practice, it is not computationally feasible to consider all permutations even when $N$ is moderately large. In such scenario, we generate random permutations $\pi_1,\pi_2\ldots,\pi_B$ of the set $\{1,2\ldots, N\}$ and obtain a randomized p-value
$$p_{n,m,B} = \frac{1}{B+1}\big\{\sum_{i=1}^B\mathbbm{1}[\hat\zeta_{n,m}^{\phi,\pi_i}\geq \hat\zeta_{n,m}^{\phi}]+1\big\}.$$
We have seen that the use all the $N!$ permutations leads to the p-value
$$p_{n,m} = \frac{1}{N!}\Big\{\sum_{\pi\in \mathcal{S}_N}\mathbbm{1}[\hat\zeta_{n,m}^{\phi,\pi}\geq \hat\zeta_{n,m}^{\phi}]\Big\},$$
Naturally, one would expect $p_{n,m,B}$ and $p_{n,m}$ to be close as the number of random permutations $B$ grows to infinity. This is asserted by the following proposition.

\begin{prop}
For any given $\mathcal{U}$, $p_{n,m,B}$ converges almost surely to $p_{n,m}$ as $B$ grows to infinity.
\label{MCpval}
\end{prop}


So, when $B$ and $\min\{n,m\}$ are sufficiently large, the randomized permutation test (which is used in practice) approximates the oracle test, i.e., the test that assumes the knowledge of the underlying data generating distributions. Such knowledge is generally never available. While the randomized permutation test can be applied without such knowledge. This strongly advocates the use of the randomized permutation test in practice.

\subsection{Local asymptotic behaviour of the test}
\label{sec-loc}
In this section, we construct a locally asymptotically normal sequence of contiguous alternatives and study the behaviour of our test under such alternatives. 
Suppose that $Z_1,Z_2,\ldots,Z_N$ are independent and identically distributed functional random variables with distribution $F$. 
Define $F^{(N)} = (1-\delta_N) F+\delta_N L$, where $F$ and $L$ are two probability distributions on $\H$, and $\{\delta_N\}$ is a sequence in $(0,1)$ that converges to zero as $N$ grows to infinity. 
Clearly, the total variation distance between $F^{(N)}$ and $F$ converges to zero as $N$ diverges to infinity. Hence, $F^{(N)}$ and $F$ are mutually contiguous for any probability distribution $L$ and a sequence $\{\delta_N\}$ in $(0,1)$ that converges to zero as $N$ increases.
Now, for studying the local behavior of our test, we assume that
\begin{itemize}
    \item[(A1)] $L$ is absolutely continuous with respect to $F$ with square integrable density $\ell(.)$. 
\end{itemize} 

Under assumption (A1), $F^{(N)}$ is absolutely continuous with respect to $F$ and for $\delta_N = \alpha/\sqrt{N}$, we have the following result on the local asymptotic normality for functional random variables.
\begin{thm}
    Under assumption (A1) and $\delta_N = \alpha/\sqrt{N}$, the Radon-Nikodym derivative of $F^{(N)}$ with respect to $F$ is 
    $\Big(1+\frac{\alpha}{\sqrt{N}}\big(\ell(z)-1\big)\Big)$, and as $N$ goes to infinity, we have
    $$\left|\log\Big\{\prod_{i=1}^N\frac{dF^{(N)}}{dF}(Z_i)\Big\}-\frac{\alpha}{\sqrt{N}}\sum_{i=1}^N \Big(\ell(Z_i)-1\Big)+\frac{\alpha^2}{2}\E\Big\{\ell(Z_1)-1\Big\}^2\right|\stackrel{P}{\rightarrow} 0,$$
    \label{lanfd}
\vspace{-0.1in}
\end{thm}
Since for functional random variables, there is no universally accepted dominating measure such as the Lebesgue measure, the quadratic mean differentiability assumption is quite difficult to formulate. But contiguity through contamination alternatives is naturally extendable in such cases.  
Let $F^{(n)} =  F$ and $G^{(m)} = (1-\delta/\sqrt{m})F+\delta/\sqrt{m} L$ for some probability distribution $L$ satisfying (A1) and a positive number $\delta$. Clearly, the alternative $(F^{(n)},G^{(m)})$ is contiguous with the null $(F,F)$. The next theorem shows that under $(F^{(n)},G^{(m)})$, $nm/(n+m)\hat\zeta_{n,m}^\phi$ converges in distribution to a tight random variable.

\begin{thm}
Under $(F^{(n)},G^{(m)})$, as $\min\{n,m\}$ grows to infinity, $nm/(n+m)\hat\zeta_{n,m}^{\phi}$ converges in distribution to $\sum_{k=1}^\infty \lambda_k \big(\xi_k+\sqrt{\lambda}\delta (\int \varphi_kdL-\int \varphi_kdF)\big)^2$, where $\lim n/(n+m)=\lambda\in (0,1)$, $\{\xi_k\}$ is a sequence of i.i.d. standard normal random variables and $\{\lambda_k\}$ and $\{\varphi_k\}$ are the eigenvalues and eigenfunctions of the integral equation $\int h(u,v)\gamma(v)dF(v) = \lambda \gamma(u)$, where for $U_1,U_2,U_3\stackrel{i.i.d.}{\sim}H$, $h(u,v) = \E\Big\{\phi\big(|\langle u,U_1\rangle, \langle U_2, U_1\rangle|^2\big)\Big\} + \E\Big\{\phi\big(|\langle v,U_1\rangle, \langle U_2, U_1\rangle|^2\big)\Big\} - 2\E\Big\{\phi\big(|\langle u,U_1\rangle - \langle v, U_1\rangle|^2\big)\Big\}.$
    \label{local-limit}
\end{thm}

To study the asymptotic behavior of the test, we also need to study the convergence of the permuted statistic $\hat\zeta_{n,m}^{\phi,\pi}$ under the sequence of alternatives $(F^{(n)},G^{(m)})$. This is given by the following theorem.

\begin{thm}
 If $\phi$ satisfies the conditions mentioned in Proposition 2.1, 
 under the contiguous alternative $(F^{(n)},G^{(m)})$, as $\min\{n,m\}$ grows to infinity, 
 $nm/(n+m)\hat\zeta_{n,m}^{\phi,\pi}$ converges in distribution to $\sum_{k=1}^\infty \lambda_k Z_k^2$ where $\{Z_k\}$ and $\{\lambda_k\}$ are as in Theorem \ref{largesampleres}.
 \label{local-limit-per-2}
\end{thm}

Theorems \ref{local-limit} and \ref{local-limit-per-2} together show that for a suitable choice of $\phi$, under $(F^{(n)},G^{(m)})$ the power of our test converges to a non-trivial limit which is a function of $\delta$. This shows that the proposed test is statistically efficient in the Pitman sense. It can be easily verified that as $\delta$ diverges to infinity the asymptotic power will be unity and it will be equal to the level $\alpha$ when $\delta$ shrinks to zero. The exact expression of the limit is not analytically tractable. So, in Section \ref{emp-efficiency} we compare the efficiency different tests through simulations.  

\section{Empirical performance of the proposed test}

In this section, we evaluate the empirical performance of our test by carrying out some simulated experiments and analyzing a real dataset. Note that in practice, one needs to choose a suitable function $\phi$ to implement the test. There are several choice of $\phi$ available in the literature, but here we take the functions $\phi_1(z) = \sqrt{z}/2$, $\phi_2(z) = 1-\exp(-z/2)$ and $\phi_3(z) = \log(1+z)$ to construct our test, which we further refer to as pBF-L2, pBF-exp and pBF-log tests.

First, we consider some examples for studying the level property of our tests and then we compare their powers with the powers of the tests proposed in \cite{pomann2016two}, \cite{wynne2020kernel} and \cite{pan2018ball}, which are referred to as the FAD (Functional Anderson-Darling) test, WD (Wynne-Duncan) test, and BD (Ball Divergence) test respectively. In our simulated experiments, 
the randomized p-values of the permutation tests are computed based on 500 random permutations and each experiment is repeated 1000 times to estimate the power of a test by the proportion of times it rejected $H_0$. Throughout this section, all tests are considered to have 5\% nominal level.

\subsection{Analysis of simulated data sets}
First, we study the level properties of our tests. For this purpose, we generate the samples $\mathcal{X}$ and $\mathcal{Y}$ from the same distribution. Here we consider three examples (Example 1-3) and compute the power (which is the same as the level when $F=G$) for different sample sizes ($n=m=20,30,40$ and $50$). 

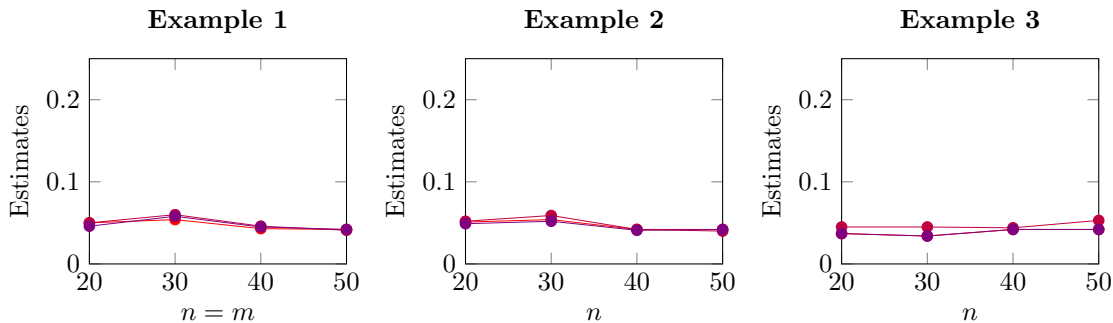
\begin{figure}[!h]
\centering
\begin{tikzpicture}
\begin{axis}[xmin = 20, xmax = 50, ymin = 0, ymax = 0.25, xlabel = {$n=m$}, ylabel = {Estimates}, title = {\bf Example 1}]
\addplot[color = red,   mark = *, step = 1cm,very thin]coordinates{(20,0.05)(30,0.054)(40,0.043)(50,0.042)};

\addplot[color = purple,   mark = *, step = 1cm,very thin]coordinates{(20,0.05)(30,0.06)(40,0.046)(50,0.041)};

\addplot[color = violet,   mark = *, step = 1cm,very thin]coordinates{(20,0.046)(30,0.058)(40,0.045)(50,0.042)};

\end{axis}
\end{tikzpicture}
\begin{tikzpicture}
\begin{axis}[xmin = 20, xmax = 50, ymin = 0, ymax = 0.25, xlabel = {$n$}, ylabel = {Estimates}, title = {\bf Example 2}]
\addplot[color = red,   mark = *, step = 1cm,very thin]coordinates{(20,0.051)(30,0.054)(40,0.042)(50,0.042)};

\addplot[color = purple,   mark = *, step = 1cm,very thin]coordinates{(20,0.052)(30,0.059)(40,0.042)(50,0.04)};

\addplot[color = violet,   mark = *, step = 1cm,very thin]coordinates{(20,0.049)(30,0.052)(40,0.041)(50,0.042)};

\end{axis}
\end{tikzpicture}
\begin{tikzpicture}
\begin{axis}[xmin = 20, xmax = 50, ymin = 0, ymax = 0.25, xlabel = {$n$}, ylabel = {Estimates}, title = {\bf Example 3}]
\addplot[color = red,   mark = *, step = 1cm,very thin]coordinates{(20,0.037)(30,0.034)(40,0.042)(50,0.042)};

\addplot[color = purple,   mark = *, step = 1cm,very thin]coordinates{(20,0.045)(30,0.045)(40,0.044)(50,0.053)};

\addplot[color = violet,   mark = *, step = 1cm,very thin]coordinates{(20,0.037)(30,0.034)(40,0.042)(50,0.042)};

\end{axis}
\end{tikzpicture}
    
    \caption{Results of pBF-L2 (\textcolor{red}{$\tikzcircle{2pt}$}), pBF-exp (\textcolor{violet}{$\tikzcircle{2pt}$}) and pBF-log (\textcolor{purple}{$\tikzcircle{2pt}$}) tests for $n=m=20,30,40$ and $50$ in Examples 1-3.}
    \label{fig:level}
\end{figure}

\vspace{0.1in}
\noindent
\textbf{Example 1} $X$ and $Y$ are independent Wiener process $W$ on $[0,1]$.
\vspace{0.1in}

\noindent
\textbf{Example 2} $X$ and $Y$ are independently distributed as $\mu+W$ on $[0,1]$ where $\mu(t) = t$ and $W$ is the Wiener process. 
\vspace{0.1in}

\noindent
\textbf{Example 3} $X$ and $Y$ are independent random functions defined as $\sum_{i=1}^9 \frac{1}{i^{2.5}}\xi_i \psi_i(t)$, where $\xi_i$ are i.i.d. $N(0,1)$ random variables and $\{\psi_i\}$ is the trigonometric basis on $L_2([0,1])$.
\vspace{0.1in}

Figure \ref{fig:level} shows that the observed level of our tests is approximately $0.05$ in all three examples, which we expect in view of the theoretical results stated in the previous sections. 

Next, we consider some location and scale alternatives (Examples 4 and 5) to compare the power of pBF-L2, pBF-exp, and pBF-log tests with that of FAD, WD, and BD tests. 

\vspace{0.1in}
\noindent
\textbf{Example 4} $X$ is the Wiener process $W$ on $[0,1]$ while $Y$ is distributed as $\mu+W$ and is independent of $X$. We consider two choices of $\mu$, (i) $\mu(t) = r t^2$ and (ii) $\mu(t) = r e^t$, and carry out our experiment for different choices of $r\in \mathbb [0,1]$ as shown in Figure~\ref{fig:5.4} .

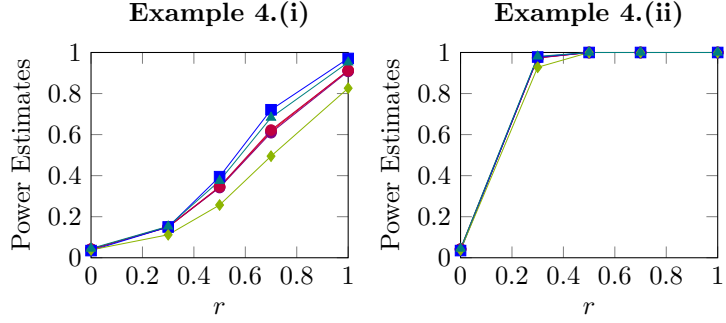
\begin{figure}[h!]
\centering








\begin{tikzpicture}
\begin{axis}[xmin = 0, xmax = 1, ymin = 0, ymax = 1, xlabel = {$r$}, ylabel = {Power Estimates}, title = {\bf Example 4.(i)}]
\addplot[color = red, mark = *, step = 1cm,very thin]coordinates{(0,0.042)(0.3,0.149)(0.5,0.345)(0.7,0.616)(1,0.913)};

\addplot[color = violet, mark = *, step = 1cm,very thin]coordinates{(0,0.042)(0.3,0.15)(0.5,0.342)(0.7,0.609)(1,0.908)};

\addplot[color = purple, mark = *, step = 1cm,very thin]coordinates{(0,0.041)(0.3,0.151)(0.5,0.346)(0.7,0.622)(1,0.912)};

\addplot[color = blue, mark = square*, step = 1cm,very thin]coordinates{(0,0.035)(0.3,0.15)(0.5,0.395)(0.7,0.721)(1,0.971)};

\addplot[color = applegreen, mark = diamond*, step = 1cm,very thin]coordinates{(0,0.038)(0.3,0.112)(0.5,0.257)(0.7,0.495)(1,0.826)};

\addplot[color = teal, mark = triangle*, step = 1cm,very thin]coordinates{(0,0.045)(0.3,0.153)(0.5,0.375)(0.7,0.683)(1,0.952)};
\end{axis}
\end{tikzpicture}
\begin{tikzpicture}
\begin{axis}[xmin = 0, xmax = 1, ymin = 0, ymax = 1, xlabel = {$r$}, ylabel = {Power Estimates}, title = {\bf Example 4.(ii)}]
\addplot[color = red, mark = *, step = 1cm,very thin]coordinates{(0,0.042)(0.3,0.975)(0.5,1)(0.7,1)(1,1)};

\addplot[color = violet, mark = *, step = 1cm,very thin]coordinates{(0,0.042)(0.3,0.973)(0.5,1)(0.7,1)(1,1)};

\addplot[color = purple, mark = *, step = 1cm,very thin]coordinates{(0,0.041)(0.3,0.974)(0.5,1)(0.7,1)(1,1)};

\addplot[color = blue, mark = square*, step = 1cm,very thin]coordinates{(0,0.035)(0.3,0.978)(0.5,1)(0.7,1)(1,1)};

\addplot[color = applegreen, mark = diamond*, step = 1cm,very thin]coordinates{(0,0.038)(0.3,0.928)(0.5,1)(0.7,1)(1,1)};

\addplot[color = teal, mark = triangle*, step = 1cm,very thin]coordinates{(0,0.045)(0.3,0.981)(0.5,1)(0.7,1)(1,1)};
\end{axis}
\end{tikzpicture}
    
    \caption{Results of pBF-L2 test (\textcolor{red}{$\bullet$}), pBF-exp test (\textcolor{violet}{$\bullet$}), pBF-log test (\textcolor{purple}{$\bullet$}), FAD test (\textcolor{blue}{$\blacksquare$}), BD test (\textcolor{applegreen}{$\blacklozenge$}) and WD test (\textcolor{teal}{$\blacktriangle$}) for Example 4.(i) and (ii).}
    \label{fig:5.4}
\end{figure}
\vspace{0.1in}
In this example, we generate 50 observations on each of the random variables $X$ and $Y$. Here $X$ is pure noise, whereas $Y$ has a non-zero signal $\mu$. The location difference between the two distributions is an increasing function of $r$. So, as expected, the powers of all tests are increasing with $r$ (see Figure \ref{fig:5.4}). In both scenarios, the pBF tests are competitive with the FAD and WD tests whereas the BD test had a relatively poor performance.

\vspace{0.1in}
\noindent
\textbf{Example 5} $X$ and $Y$ are independent random functions as in Example 3 with respective coefficients denoted as $\xi_i^X$ and $\xi_i^Y$ for each $i=1,2,\ldots,9$. Here we consider two scale problems:
(i) $\xi_i^X$s are i.i.d $N(0,1)$ random variable, while $\xi_i^Y$s are i.i.d. $N(0,\sigma^2)$ random variables; (ii) $\xi_i^X$s are independent standard Cauchy variables and $\xi_i^Y$s are the centered Cauchy random variables with scale parameters $\sigma>0$.

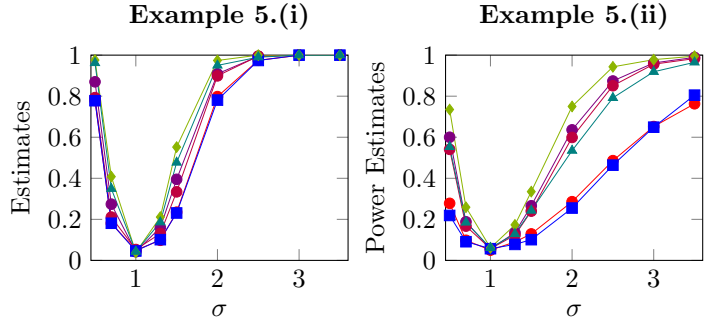
\begin{figure}[h!]
\centering
\begin{tikzpicture}
\begin{axis}[xmin = 0.45, xmax = 3.6, ymin = 0, ymax = 1, xlabel = {$\sigma$}, ylabel = {Estimates}, title = {\bf Example 5.(i)}]
\addplot[color = red, mark = *, step = 1cm,very thin]coordinates{(0.5,0.779)(0.7,0.181)(1,0.05)(1.3,0.098)(1.5,0.235)(2,0.798)(2.5,0.975)(3,0.999)(3.5,1)};

\addplot[color = violet, mark = *, step = 1cm,very thin]coordinates{(0.5,0.87)(0.7,0.274)(1,0.042)(1.3,0.16)(1.5,0.395)(2,0.909)(2.5,0.994)(3,1)(3.5,1)};

\addplot[color = purple, mark = *, step = 1cm,very thin]coordinates{(0.5,0.793)(0.7,0.211)(1,0.053)(1.3,0.131)(1.5,0.334)(2,0.899)(2.5,0.995)(3,1)(3.5,1)};

\addplot[color = blue, mark = square*, step = 1cm,very thin]coordinates{(0.5,0.778)(0.7,0.181)(1,0.046)(1.3,0.101)(1.5,0.231)(2,0.781)(2.5,0.974)(3,1)(3.5,1)};

\addplot[color = applegreen, mark = diamond*, step = 1cm,very thin]coordinates{(0.5,0.976)(0.7,0.408)(1,0.039)(1.3,0.211)(1.5,0.552)(2,0.974)(2.5,1)(3,1)(3.5,1)};

\addplot[color = teal, mark = triangle*, step = 1cm,very thin]coordinates{(0.5,0.962)(0.7,0.349)(1,0.043)(1.3,0.186)(1.5,0.477)(2,0.95)(2.5,0.99)(3,1)(3.5,1)};

\end{axis}
\end{tikzpicture}
\begin{tikzpicture}
\begin{axis}[xmin = 0.45, xmax = 3.6, ymin = 0, ymax = 1, xlabel = {$\sigma$}, ylabel = {Power Estimates}, title = {\bf Example 5.(ii)}]
\addplot[color = red, mark = *, step = 1cm,very thin]coordinates{(0.5,0.278)(0.7,0.098)(1,0.05)(1.3,0.09)(1.5,0.129)(2,0.286)(2.5,0.486)(3,0.652)(3.5,0.763)};

\addplot[color = violet, mark = *, step = 1cm,very thin]coordinates{(0.5,0.6)(0.7,0.188)(1,0.056)(1.3,0.135)(1.5,0.266)(2,0.636)(2.5,0.874)(3,0.962)(3.5,0.989)};

\addplot[color = purple, mark = *, step = 1cm,very thin]coordinates{(0.5,0.541)(0.7,0.167)(1,0.057)(1.3,0.123)(1.5,0.24)(2,0.599)(2.5,0.851)(3,0.956)(3.5,0.983)};

\addplot[color = blue, mark = square*, step = 1cm,very thin]coordinates{(0.5,0.219)(0.7,0.091)(1,0.057)(1.3,0.078)(1.5,0.101)(2,0.255)(2.5,0.464)(3,0.649)(3.5,0.805)};

\addplot[color = applegreen, mark = diamond*, step = 1cm,very thin]coordinates{(0.5,0.735)(0.7,0.259)(1,0.061)(1.3,0.173)(1.5,0.335)(2,0.75)(2.5,0.943)(3,0.978)(3.5,0.995)};

\addplot[color = teal, mark = triangle*, step = 1cm,very thin]coordinates{(0.5,0.553)(0.7,0.185)(1,0.061)(1.3,0.129)(1.5,0.243)(2,0.535)(2.5,0.792)(3,0.919)(3.5,0.965)};
\end{axis}
\end{tikzpicture}
  
\caption{Results of pBF-L2 test (\textcolor{red}{$\bullet$}), pBF-exp test (\textcolor{violet}{$\bullet$}), pBF-log test (\textcolor{purple}{$\bullet$}), FAD test (\textcolor{blue}{$\blacksquare$}), BD test (\textcolor{applegreen}{$\blacklozenge$}) and WD test (\textcolor{teal}{$\blacktriangle$}) for Examples 5 (i) and (ii).}
\label{fig:5.5}
\end{figure}

\vspace{0.1in}
Here also we generate 50 observations on each of the random variables $X$ and $Y$. Note, as $\sigma$ deviates from 1, since the scale difference between the two distributions increases, the power of a test is also expected to increase. Figure \ref{fig:5.5} shows the powers of different tests. Here the BD test has the best performance. In 5.(i) the power of the WD, pBF-exp, and pBF-log tests had similar performance, but in 5.(ii) pBF-exp and pBF-log tests outperforms the WD test. The pBF-L2 and FAD tests had a relatively poor performance in this example.

Next, we consider an examples (Example 6 and 7) where the projection-based tests have superior performance over the distance-based tests.


\vspace{0.1in}
\textbf{Example 6} Define $X(t) = \sum_{i=1}^d \frac{1}{\sqrt{d}}\xi_i \sqrt{2}\sin(2\pi it)$ and $Y(t) = \sum_{i=1}^d \frac{1}{\sqrt{d}}\eta_i \cos(2\pi it)$, where the $\xi_i$s and the $\eta_i$s are independent mean zero random variables with the same variance. Here we consider two cases: (i) $\xi_i$s and $\eta_i$s are i.i.d. $N(0,2)$ random variables with mean $0$ variance $2$, (ii) $\xi_i$s are i.i.d $N(0,2)$ but $\eta_i$s follow the standard t-distribution with $4$ degrees of freedom.

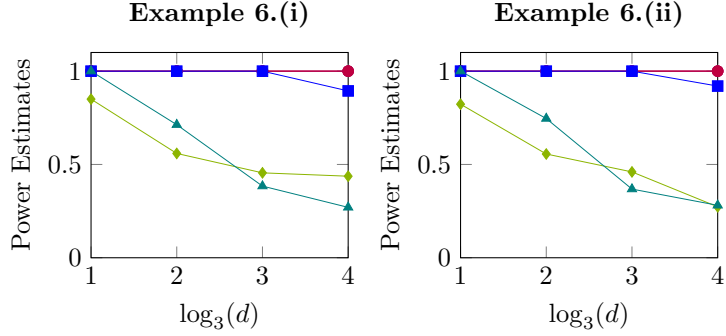
\begin{figure}[!h]
\centering
\begin{tikzpicture}
\begin{axis}[xmin = 1, xmax = 4, ymin = 0, ymax = 1.1, xlabel = {$\log_3(d)$}, ylabel = {Power Estimates}, title = {\bf Example 6.(i)}]
\addplot[color = red, mark = *, step = 1cm,very thin]coordinates{(1,1)(2,1)(3,1)(4,1)};

\addplot[color = violet, mark = *, step = 1cm,very thin]coordinates{(1,1)(2,1)(3,1)(4,1)};

\addplot[color = purple, mark = *, step = 1cm,very thin]coordinates{(1,1)(2,1)(3,1)(4,1)};

\addplot[color = blue, mark = square*, step = 1cm,very thin]coordinates{(1,1)(2,1)(3,1)(4,0.894)};

\addplot[color = applegreen, mark = diamond*, step = 1cm,very thin]coordinates{(1,0.85)(2,0.559)(3,0.455)(4,0.437)};

\addplot[color = teal, mark = triangle*, step = 1cm,very thin]coordinates{(1,1)(2,0.713)(3,0.384)(4,0.27)};
\end{axis}
\end{tikzpicture}
\begin{tikzpicture}
\begin{axis}[xmin = 1, xmax = 4, ymin = 0, ymax = 1.1, xlabel = {$\log_3(d)$}, ylabel = {Power Estimates}, title = {\bf Example 6.(ii)}]
\addplot[color = red, mark = *, step = 1cm,very thin]coordinates{(1,1)(2,1)(3,1)(4,1)};

\addplot[color = violet, mark = *, step = 1cm,very thin]coordinates{(1,1)(2,1)(3,1)(4,1)};

\addplot[color = purple, mark = *, step = 1cm,very thin]coordinates{(1,1)(2,1)(3,1)(4,1)};

\addplot[color = blue, mark = square*, step = 1cm,very thin]coordinates{(1,1)(2,1)(3,1)(4,0.92)};

\addplot[color = applegreen, mark = diamond*, step = 1cm,very thin]coordinates{(1,0.823)(2,0.556)(3,0.46)(4,0.273)};

\addplot[color = teal, mark = triangle*, step = 1cm,very thin]coordinates{(1,1)(2,0.746)(3,0.368)(4,0.281)};
\end{axis}
\end{tikzpicture}
    \caption{Results of pBF-L2 test (\textcolor{red}{$\bullet$}), pBF-exp test (\textcolor{violet}{$\bullet$}), pBF-log test (\textcolor{purple}{$\bullet$}), FAD test (\textcolor{blue}{$\blacksquare$}), BD test (\textcolor{applegreen}{$\blacklozenge$}) and WD test (\textcolor{teal}{$\blacktriangle$}) for Examples 6.(i) and 6.(ii)}
    \label{fig:6}
\end{figure}

\vspace{0.1in}

We generate $30$ observations on both $X$ and $Y$ and compute the power of the tests for $d=3^i$ with $i=1,2,\ldots, 4$. Figure \ref{fig:6} shows that the power of the pBF-L2, pBF-exp, and pBF-log tests remained unity for all choices of $d$. The FAD test also performed well, but our tests had an edge for large $d$. Here the BD and WD tests had a decaying performance with increasing $d$.

Note that in this example, the squared pairwise distances are $\|X_1-Y_1\|^2 = \sum_{i=1}^d \xi_i^2/d + \sum_{i=1}^d \eta_i^2/d$, $\|X_1-X_2\|^2 = \sum_{i=1}^d (\xi_{1i}-\xi_{2i})^2/d$ and $\|Y_1-Y_2\|^2 = \sum_{i=1}^d (\eta_{1i}-\eta_{2i})^2/d$. 
One can show that with increasing $d$ the pairwise distances converges to the same limit. Therefore if the distributions of $\xi_i$ and $\eta_i$ are close, it is expected that the distributions of $\|X_1-Y_1\|, \|X_1-X_2\|$ and $\|Y_1-Y_2\|$ will also be close for large $d$. On the other hand, we have $\langle X_1,X_2 \rangle = \sum_{i=1}^d \xi_{1i}\xi_{2i}/d$, $\langle Y_1,Y_2 \rangle = \sum_{i=1}^d \eta_{1i}\eta_{2i}/d$ but $\langle X_1,Y_1 \rangle = 0$. 
The pBF tests successfully discriminate between non-degenerate and degenerate distributions (i.e. the distributions of $\langle X_1, X_2 \rangle$ and $\langle Y_1, X_2\rangle$ or 
the distributions of $\langle Y_1, Y_2\rangle$ and  $\langle X_1, Y_2 \rangle$) and aggregate them in a suitable way.
As a result, they had superior performance compared to the distance-based methods. From the above discussion, one would expect these tests to exhibit a similar behaviour if $\phi_{0,i}(t) = \sqrt{2}\sin(2\pi it)$ and $\psi_{0,i}(t) = \sqrt{2}\cos(2\pi it)$ are replaced by some other orthonormal sequences $\phi_i(t)$ and $\psi_i(t)$ which are orthogonal among themselves. 

\textbf{Example 7} Define $U(t) = \sum_{i=1}^{81} \frac{1}{\sqrt{81}}\xi_i \sqrt{2}\sin(2\pi it)$ and $V(t) = \sum_{i=1}^{81} \frac{1}{\sqrt{81}}\eta_i \cos(2\pi it)$. Sample $X_1,X_2,\ldots, X_n$ from the distribution of $U$ and $Y_1,Y_2,\ldots, Y_n$ from the mixture distribution of $U$ and $V$ with mixing proportion $\delta/\sqrt{n}$ for some $\delta>0$.

Here the alternate distribution does not satisfy assumption (A1) from Section \ref{sec-loc}. Hence, we can not claim that for large $n$ the power serves as an approximation of the efficiency of the tests. But this alternative is a contiguous alternative and we compare the power of the tests for different values of $\delta$ and different sample sizes. In Figure \ref{fig:ex7} we see that the pBF-L2 test has an overwhelming performance over the other tests for every choice of $\delta$. The pBF-log and pBF-exp also have a good performance followed by FAD test. The BD and WD test have a very poor performance of every choice of $\delta$. 

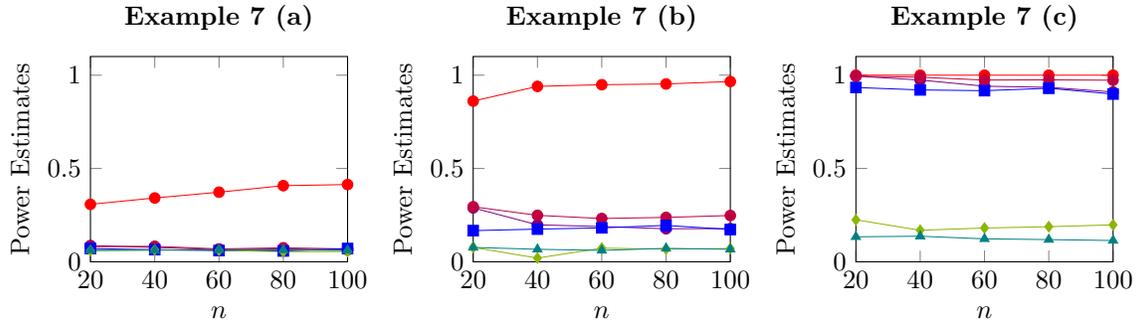
\begin{figure}[!h]
\centering
\begin{tikzpicture}
\begin{axis}[xmin = 20, xmax = 100, ymin = 0, ymax = 1.1, xlabel = {$n$}, ylabel = {Power Estimates}, title = {\bf Example 7 (a)}]
\addplot[color = red, mark = *, step = 1cm,very thin]coordinates{(20,0.308)(40,0.342)(60,0.373)(80,0.408)(100,0.414)};

\addplot[color = violet, mark = *, step = 1cm,very thin]coordinates{(20,0.087)(40,0.077)(60,0.065)(80,0.069)(100,0.064)};

\addplot[color = purple, mark = *, step = 1cm,very thin]coordinates{(20,0.081)(40,0.083)(60,0.069)(80,0.075)(100,0.069)};

\addplot[color = blue, mark = square*, step = 1cm,very thin]coordinates{(20,0.071)(40,0.065)(60,0.06)(80,0.058)(100,0.071)};

\addplot[color = applegreen, mark = diamond*, step = 1cm,very thin]coordinates{(20,0.06)(40,0.064)(60,0.062)(80,0.054)(100,0.055)};

\addplot[color = teal, mark = triangle*, step = 1cm,very thin]coordinates{(20,0.061)(40,0.062)(60,0.064)(80,0.064)(100,0.061)};
\end{axis}
\end{tikzpicture}
\begin{tikzpicture}
\begin{axis}[xmin = 20, xmax = 100, ymin = 0, ymax = 1.1, xlabel = {$n$}, ylabel = {Power Estimates}, title = {\bf  Example 7 (b)}]
\addplot[color = red, mark = *, step = 1cm,very thin]coordinates{(20,0.861)(40,0.94)(60,0.949)(80,0.953)(100,0.966)};

\addplot[color = violet, mark = *, step = 1cm,very thin]coordinates{(20,0.289)(40,0.198)(60,0.189)(80,0.177)(100,0.178)};

\addplot[color = purple, mark = *, step = 1cm,very thin]coordinates{(20,0.295)(40,0.249)(60,0.232)(80,0.238)(100,0.248)};

\addplot[color = blue, mark = square*, step = 1cm,very thin]coordinates{(20,0.167)(40,0.175)(60,0.182)(80,0.195)(100,0.173)};

\addplot[color = applegreen, mark = diamond*, step = 1cm,very thin]coordinates{(20,0.077)(40,0.02)(60,0.073)(80,0.068)(100,0.071)};

\addplot[color = teal, mark = triangle*, step = 1cm,very thin]coordinates{(20,0.077)(40,0.067)(60,0.061)(80,0.072)(100,0.067)};
\end{axis}
\end{tikzpicture}
\begin{tikzpicture}
\begin{axis}[xmin = 20, xmax = 100, ymin = 0, ymax = 1.1, xlabel = {$n$}, ylabel = {Power Estimates}, title = {\bf  Example 7 (c)}]
\addplot[color = red, mark = *, step = 1cm,very thin]coordinates{(20,1)(40,1)(60,1)(80,1)(100,1)};

\addplot[color = violet, mark = *, step = 1cm,very thin]coordinates{(20,0.995)(40,0.974)(60,0.941)(80,0.935)(100,0.91)};

\addplot[color = purple, mark = *, step = 1cm,very thin]coordinates{(20,0.995)(40,0.989)(60,0.974)(80,0.975)(100,0.973)};

\addplot[color = blue, mark = square*, step = 1cm,very thin]coordinates{(20,0.934)(40,0.921)(60,0.917)(80,0.93)(100,0.899)};

\addplot[color = applegreen, mark = diamond*, step = 1cm,very thin]coordinates{(20,0.225)(40,0.169)(60,0.181)(80,0.188)(100,0.198)};

\addplot[color = teal, mark = triangle*, step = 1cm,very thin]coordinates{(20,0.134)(40,0.137)(60,0.123)(80,0.119)(100,0.114)};
\end{axis}
\end{tikzpicture}

    \caption{Results of pBF-L2 test (\textcolor{red}{$\bullet$}), pBF-exp test (\textcolor{violet}{$\bullet$}), pBF-log test (\textcolor{purple}{$\bullet$}), FAD test (\textcolor{blue}{$\blacksquare$}), BD test (\textcolor{applegreen}{$\blacklozenge$}) and WD test (\textcolor{teal}{$\blacktriangle$}) for Examples 7 when (a) $\delta =1$, (b) $\delta = 2$ and (c) $\delta = 4$.}
    \label{fig:ex7}
\end{figure}

\subsection{Empirical efficiency of the tests}
\label{emp-efficiency}
In this section, we compare the efficiency of the above tests by generating $X_1,X_2,\ldots, X_n$ independently from a distribution $F$ and $Y_1,Y_2,\ldots, Y_n$ independently from a contiguous locally asymptotically normal alternative $(1- \delta/\sqrt{n})F+\delta/\sqrt{n} G~ (G\not=F)$ for different values of $\delta$. We take $F$ and $G$ as follows.

\textbf{Example 8} $F$ denote the distribution of $\sum_{i=1}^9 \frac{1}{i^{2.5}}\xi_{i} \psi_i(t)$ and $G$ denote the distribution of $\sum_{i=1}^9 \frac{1}{i^{2.5}}\eta_{i} \psi_i(t)$ where $\xi_i$s are i.i.d. $N(0,1)$ random variables and $\eta_{i}$s are i.i.d. $N(1,1)$ random variables.

\textbf{Example 9} $F$ denote the distribution of $\sum_{i=1}^9 \frac{1}{i^{2.5}}\xi_{i} \psi_i(t)$ and $G$ denote the distribution of $\sum_{i=1}^9 \frac{1}{i^{2.5}}\eta_{i} \psi_i(t)$ $\xi_i$s are i.i.d. $N(0,1)$ random variables and $\eta_{i}$s are i.i.d. $N(0,2)$ random variables.



Here $\{\psi_i\}$ is the Trigonometric basis of $L_2([0,1])$ as before. Figure \ref{fig:efficiency} displays the power of the tests for different values of $\delta$ and sample size $n$. Note that when $n$ is large, the power of a test serves as an approximation of the efficiency of that test. Hence the results can be interpreted as follows.

In Example 8, the FAD test is the most efficient test, closely followed by the pBF-L2 test. The rest of the tests can be arranged as pBF-log, WD, pBF-exp, and BD tests in decreasing order of efficiency. In Example 9, the BD test turns out to be the most efficient followed by WD, pBF-exp, pBF-log, pBF-L2, and FAD tests arranged in the same way. 

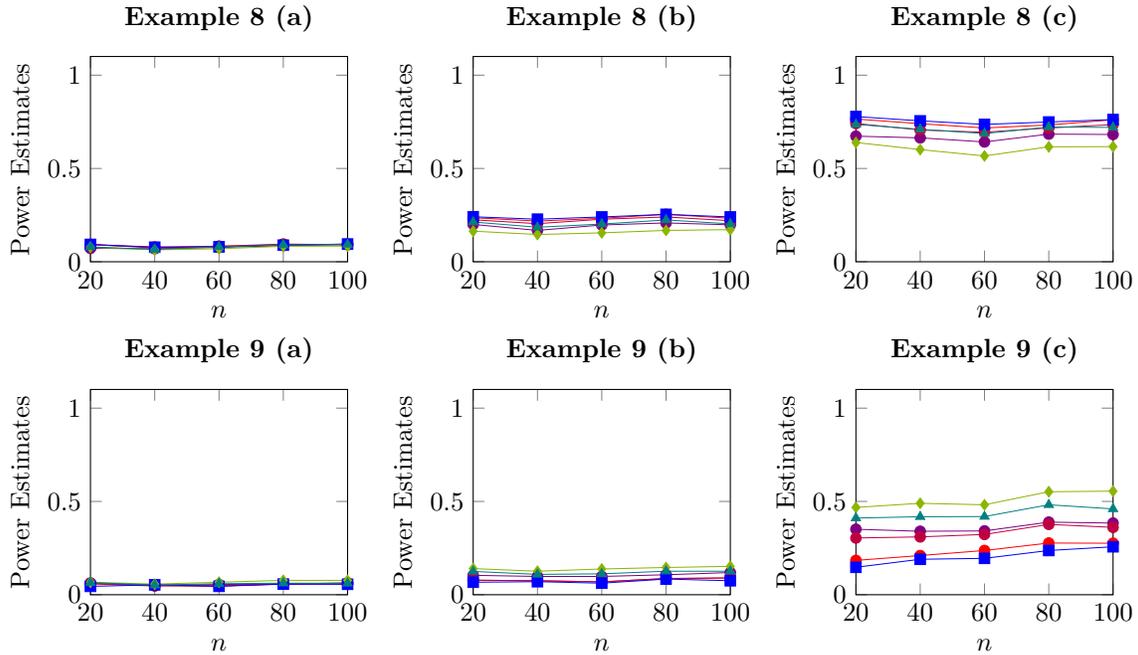
\begin{figure}[!h]
\centering
\begin{tikzpicture}
\begin{axis}[xmin = 20, xmax = 100, ymin = 0, ymax = 1.1, xlabel = {$n$}, ylabel = {Power Estimates}, title = {\bf Example 8 (a)}]
\addplot[color = red, mark = *, step = 1cm,very thin]coordinates{(20,0.091)(40,0.079)(60,0.083)(80,0.095)(100,0.091)};

\addplot[color = violet, mark = *, step = 1cm,very thin]coordinates{(20,0.072)(40,0.07)(60,0.077)(80,0.087)(100,0.092)};

\addplot[color = purple, mark = *, step = 1cm,very thin]coordinates{(20,0.093)(40,0.074)(60,0.083)(80,0.091)(100,0.094)};

\addplot[color = blue, mark = square*, step = 1cm,very thin]coordinates{(20,0.093)(40,0.078)(60,0.08)(80,0.089)(100,0.095)};

\addplot[color = applegreen, mark = diamond*, step = 1cm,very thin]coordinates{(20,0.078)(40,0.064)(60,0.07)(80,0.083)(100,0.084)};

\addplot[color = teal, mark = triangle*, step = 1cm,very thin]coordinates{(20,0.08)(40,0.066)(60,0.08)(80,0.09)(100,0.093)};
\end{axis}
\end{tikzpicture}
\begin{tikzpicture}
\begin{axis}[xmin = 20, xmax = 100, ymin = 0, ymax = 1.1, xlabel = {$n$}, ylabel = {Power Estimates}, title = {\bf Example 8 (b)}]
\addplot[color = red, mark = *, step = 1cm,very thin]coordinates{(20,0.236)(40,0.218)(60,0.233)(80,0.254)(100,0.233)};

\addplot[color = violet, mark = *, step = 1cm,very thin]coordinates{(20,0.2)(40,0.168)(60,0.197)(80,0.208)(100,0.198)};

\addplot[color = purple, mark = *, step = 1cm,very thin]coordinates{(20,0.225)(40,0.204)(60,0.228)(80,0.24)(100,0.22)};

\addplot[color = blue, mark = square*, step = 1cm,very thin]coordinates{(20,0.241)(40,0.228)(60,0.24)(80,0.254)(100,0.24)};

\addplot[color = applegreen, mark = diamond*, step = 1cm,very thin]coordinates{(20,0.164)(40,0.146)(60,0.155)(80,0.168)(100,0.172)};

\addplot[color = teal, mark = triangle*, step = 1cm,very thin]coordinates{(20,0.213)(40,0.185)(60,0.202)(80,0.224)(100,0.203)};
\end{axis}
\end{tikzpicture}
\begin{tikzpicture}
\begin{axis}[xmin = 20, xmax = 100, ymin = 0, ymax = 1.1, xlabel = {$n$}, ylabel = {Power Estimates}, title = {\bf Example 8 (c)}]
\addplot[color = red, mark = *, step = 1cm,very thin]coordinates{(20,0.765)(40,0.74)(60,0.717)(80,0.733)(100,0.761)};

\addplot[color = violet, mark = *, step = 1cm,very thin]coordinates{(20,0.673)(40,0.664)(60,0.642)(80,0.684)(100,0.682)};

\addplot[color = purple, mark = *, step = 1cm,very thin]coordinates{(20,0.741)(40,0.706)(60,0.694)(80,0.717)(100,0.735)};

\addplot[color = blue, mark = square*, step = 1cm,very thin]coordinates{(20,0.778)(40,0.755)(60,0.736)(80,0.749)(100,0.762)};

\addplot[color = applegreen, mark = diamond*, step = 1cm,very thin]coordinates{(20,0.639)(40,0.601)(60,0.567)(80,0.616)(100,0.617)};

\addplot[color = teal, mark = triangle*, step = 1cm,very thin]coordinates{(20,0.736)(40,0.71)(60,0.687)(80,0.722)(100,0.72)};
\end{axis}
\end{tikzpicture}
\begin{tikzpicture}
\begin{axis}[xmin = 20, xmax = 100, ymin = 0, ymax = 1.1, xlabel = {$n$}, ylabel = {Power Estimates}, title = {\bf Example 9 (a)}]
\addplot[color = red, mark = *, step = 1cm,very thin]coordinates{(20,0.058)(40,0.047)(60,0.046)(80,0.059)(100,0.058)};

\addplot[color = violet, mark = *, step = 1cm,very thin]coordinates{(20,0.063)(40,0.049)(60,0.055)(80,0.057)(100,0.064)};

\addplot[color = purple, mark = *, step = 1cm,very thin]coordinates{(20,0.06)(40,0.053)(60,0.048)(80,0.057)(100,0.061)};

\addplot[color = blue, mark = square*, step = 1cm,very thin]coordinates{(20,0.045)(40,0.053)(60,0.045)(80,0.057)(100,0.056)};

\addplot[color = applegreen, mark = diamond*, step = 1cm,very thin]coordinates{(20,0.067)(40,0.057)(60,0.067)(80,0.077)(100,0.077)};

\addplot[color = teal, mark = triangle*, step = 1cm,very thin]coordinates{(20,0.065)(40,0.052)(60,0.059)(80,0.062)(100,0.061)};
\end{axis}
\end{tikzpicture}
\begin{tikzpicture}
\begin{axis}[xmin = 20, xmax = 100, ymin = 0, ymax = 1.1, xlabel = {$n$}, ylabel = {Power Estimates}, title = {\bf Example 9 (b)}]
\addplot[color = red, mark = *, step = 1cm,very thin]coordinates{(20,0.078)(40,0.075)(60,0.069)(80,0.087)(100,0.09)};

\addplot[color = violet, mark = *, step = 1cm,very thin]coordinates{(20,0.104)(40,0.098)(60,0.098)(80,0.107)(100,0.119)};

\addplot[color = purple, mark = *, step = 1cm,very thin]coordinates{(20,0.078)(40,0.075)(60,0.069)(80,0.087)(100,0.09)};

\addplot[color = blue, mark = square*, step = 1cm,very thin]coordinates{(20,0.067)(40,0.07)(60,0.062)(80,0.084)(100,0.075)};

\addplot[color = applegreen, mark = diamond*, step = 1cm,very thin]coordinates{(20,0.14)(40,0.126)(60,0.138)(80,0.146)(100,0.152)};

\addplot[color = teal, mark = triangle*, step = 1cm,very thin]coordinates{(20,0.125)(40,0.109)(60,0.112)(80,0.126)(100,0.126)};
\end{axis}
\end{tikzpicture}
\begin{tikzpicture}
\begin{axis}[xmin = 20, xmax = 100, ymin = 0, ymax = 1.1, xlabel = {$n$}, ylabel = {Power Estimates}, title = {\bf Example 9 (c)}]
\addplot[color = red, mark = *, step = 1cm,very thin]coordinates{(20,0.184)(40,0.21)(60,0.237)(80,0.277)(100,0.276)};

\addplot[color = violet, mark = *, step = 1cm,very thin]coordinates{(20,0.351)(40,0.34)(60,0.342)(80,0.389)(100,0.384)};

\addplot[color = purple, mark = *, step = 1cm,very thin]coordinates{(20,0.304)(40,0.31)(60,0.323)(80,0.377)(100,0.361)};

\addplot[color = blue, mark = square*, step = 1cm,very thin]coordinates{(20,0.147)(40,0.19)(60,0.195)(80,0.237)(100,0.257)};

\addplot[color = applegreen, mark = diamond*, step = 1cm,very thin]coordinates{(20,0.468)(40,0.49)(60,0.482)(80,0.552)(100,0.555)};

\addplot[color = teal, mark = triangle*, step = 1cm,very thin]coordinates{(20,0.411)(40,0.418)(60,0.419)(80,0.482)(100,0.46)};
\end{axis}
\end{tikzpicture}

    \caption{Results of pBF-L2 test (\textcolor{red}{$\bullet$}), pBF-exp test (\textcolor{violet}{$\bullet$}), pBF-log test (\textcolor{purple}{$\bullet$}), FAD test (\textcolor{blue}{$\blacksquare$}), BD test (\textcolor{applegreen}{$\blacklozenge$}) and WD test (\textcolor{teal}{$\blacktriangle$}) for Examples 8 and 9 when (a) $\delta =1$, (b) $\delta = 2$ and (c) $\delta = 4$.}
    \label{fig:efficiency}
\end{figure}

\subsection{Analysis of DTI data}
\label{real data}

\begin{figure}[!h]
    \centering
    \includegraphics[scale = 0.30]{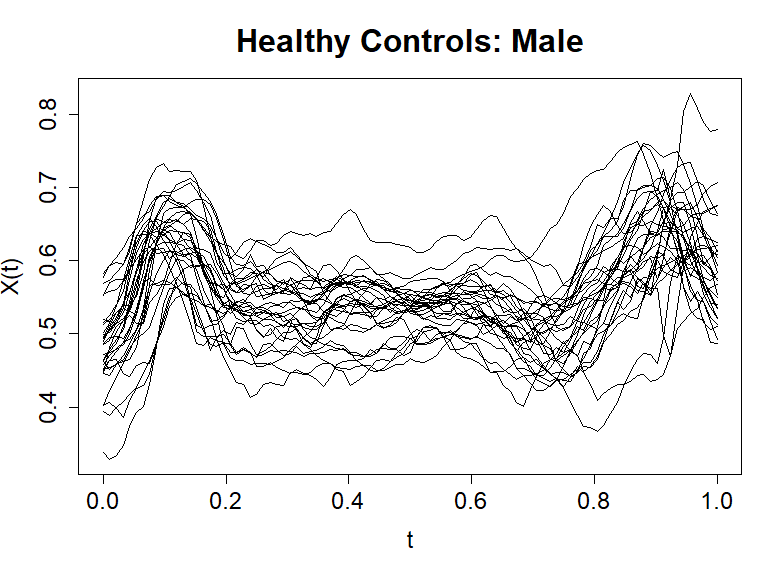}
    \includegraphics[scale = 0.30]{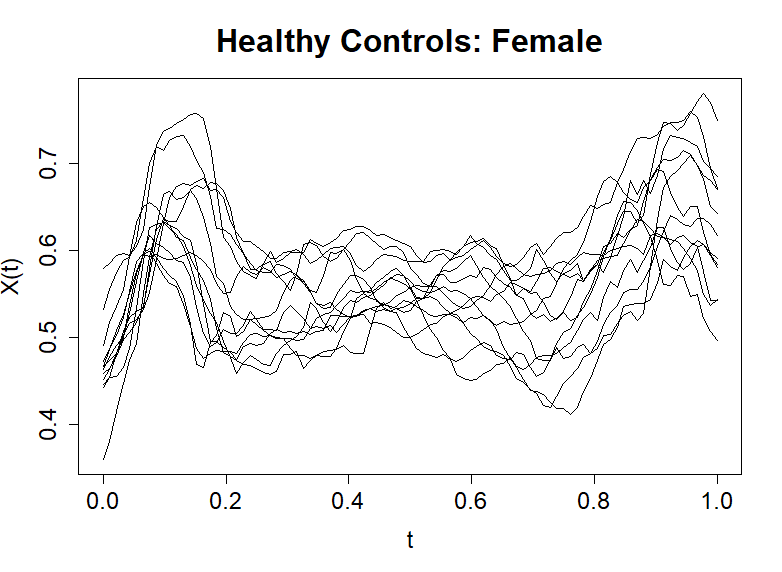}
    \includegraphics[scale = 0.30]{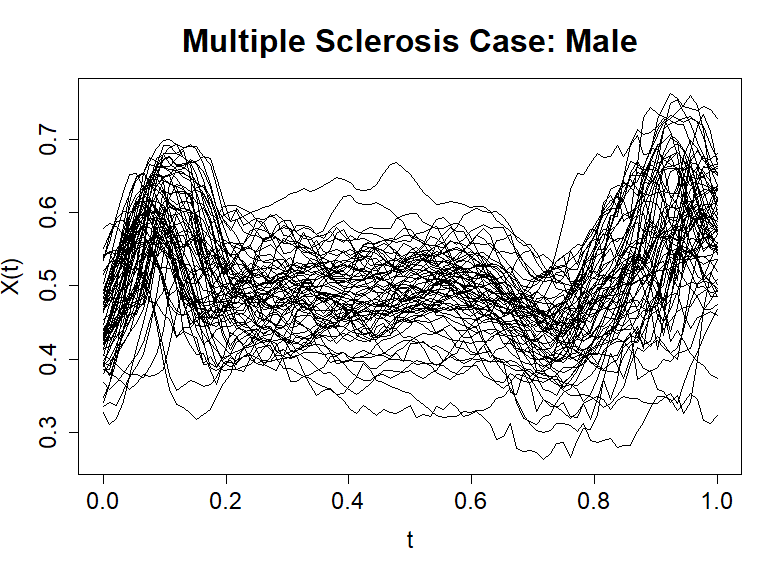}
    \includegraphics[scale = 0.30]{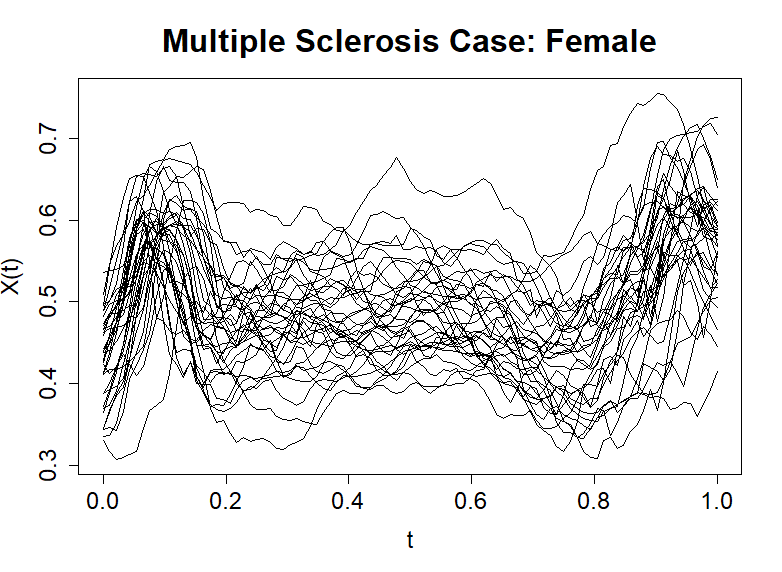}
    \caption{The FA tract profiles on the first visit, divided according to health status and gender.}
    \label{fig:dti}
\end{figure}

For further evaluation of the performance of our test, we  analyzed the DTI dataset available in the R package `refund'. The MRI/DTI data were collected at Johns Hopkins University and the Kennedy-Krieger Institute. Diffusion tensor imaging (DTI) is a magnetic resonance imaging technology that traces water diffusivity in the brain and helps to create an image of the white matter tract. This dataset has been studied in \cite{dti2011,dti2012} in the context of penalized functional regression. Several measurements of water diffusion are provided by DTI, but here we work with the fractional anisotropy (FA) tract profiles recorded at $93$ different locations of the corpus callosum. The dataset contains measurements on 100 `Multiple Sclerosis patients (MS)' and 42 `healthy controls (HC)'. While the number of visits for the MS patients ranges between 2 and 8, each healthy person visits just once. So, we consider only the FA tract profiles upon the first visit of the subjects. The subject with ID `2017' have some missing values. We delete that observation and work with the remaining 99 MS patients and 42 healthy controls. Among the subjects, there were both males and females. To test whether the `health status' or the `gender' of the subject affects the FA tract profile, we divide the dataset into four groups each corresponding to a particular combination of gender and health status. Figure \ref{fig:dti} displays the FA tracts divided into these four groups. 

Taking one pair of groups at a time, we test for the distributional difference. So, here we consider ${4\choose2} = 6$ cases: (C1) HC males vs. HC females, (C2) HC males vs. MS males, (C3) HC males vs. MS females, (C4) HC females vs. MS males, (C5) HC females vs. MS females and (C6) MS males vs. MS females.  The p-values of pBF-L2, pBF-exp, pBF-log, BD, WD tests, and Bonferonni corrected p-value of FAD test are reported in Table \ref{tab:pval-table}. Here the randomized p-values are computed based on 10,000 random permutations. Note that in many cases, the WD test fails to detect the distributional difference when the others reject the null hypothesis at $5\%$ level of significance. Our analysis suggests that the distributional difference among the males and females was statistically insignificant, but the distribution of the FA tract profile differs significantly depending on the health status.

\begin{table}[h]
    \centering
    \begin{tabular}{ccccccc} \hline
   Case & pBF-L2 & pBF-exp & pBF-log & BD & WD & FAD\\ \hline
        (C1) & $0.577$ & $0.366$ & $0.377$ & $0.803$ & $0.369$ & $0.449$\\
        (C2) & $1\times 10^{-4}$ & $1\times 10^{-4}$ & $1\times 10^{-4}$ & $1\times 10^{-4}$ & $0.096$ & $7\times 10^{-6}$\\
        (C3) & $1\times 10^{-4}$ & $1\times 10^{-4}$ & $1\times 10^{-4}$ & $1\times 10^{-4}$ & $1\times 10^{-4}$ & $8\times10^{-6}$\\
        (C4) & $4\times 10^{-4}$ & $2\times 10^{-4}$ & $2\times 10^{-4}$ & $0.001$ & $0.221$ & $0.002$\\
        (C5) & $2\times 10^{-4}$ & $2\times 10^{-4}$ & $2\times 10^{-4}$ & $4\times 10^{-4}$ & $0.542$ & $5\times 10^{-4}$\\
        (C6) & $0.574$ & $0.476$ & $0.475$ & $0.638$ & $0.639$ & $0.316$\\ \hline
    \end{tabular}
    \caption{p-values of pBF-L2, pB-log, pBF-exp, BD, WD tests, and Bonferroni corrected p-value of FAD test.}
    \label{tab:pval-table}
\end{table}

Hence, we can merge the data sets corresponding to males and females and look into the DTI data divided based on health status only. Using this we compare the performance the tests by generating random sub-samples, keeping the sample proportions from the two distributions approximately same as they were in the original data. The sub-sampling procedure was repeated 1000 times to estimate the power of the tests by the proportion of times they rejected $H_0$. Figure \ref{fig:dti_power} shows that for this data set the pBF tests had a comparable performance among themselves, and they had comparatively better performance than all the other tests. The BD and FAD tests also had a relatively good performance, whereas the WD test performed poorly in this data set.

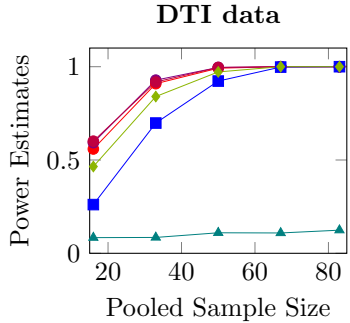
\begin{figure}[!h]
\centering
\begin{tikzpicture}
\begin{axis}[xmin = 15, xmax = 85, ymin = 0, ymax = 1.1, xlabel = {Pooled Sample Size}, ylabel = {Power Estimates}, title = {\bf DTI data}]
\addplot[color = red, mark = *, step = 1cm,very thin]coordinates{(16,0.559)(33,0.909)(50,0.992)(67,1)(83,1)};

\addplot[color = violet, mark = *, step = 1cm,very thin]coordinates{(16,0.593)(33,0.928)(50,0.996)(67,1)(83,1)};

\addplot[color = purple, mark = *, step = 1cm,very thin]coordinates{(16,0.601)(33,0.918)(50,0.996)(67,1)(83,1)};

\addplot[color = blue, mark = square*, step = 1cm,very thin]coordinates{(16,0.261)(33,0.698)(50,0.923)(67,0.998)(83,1)};

\addplot[color = applegreen, mark = diamond*, step = 1cm,very thin]coordinates{(16,0.465)(33,0.84)(50,0.972)(67,1)(83,1)};

\addplot[color = teal, mark = triangle*, step = 1cm,very thin]coordinates{(16,0.084)(33,0.085)(50,0.11)(67,0.109)(83,0.124)};
\end{axis}
\end{tikzpicture}
    \caption{Results of pBF-L2 test (\textcolor{red}{$\bullet$}), pBF-exp test (\textcolor{violet}{$\bullet$}), pBF-log test (\textcolor{purple}{$\bullet$}), FAD test (\textcolor{blue}{$\blacksquare$}), BD test (\textcolor{applegreen}{$\blacklozenge$}) and WD test (\textcolor{teal}{$\blacktriangle$}) for the DTI data.}
    \label{fig:dti_power}
\end{figure}

\section{Discussion and conclusion}

In this article, we have proposed a two-sample test for functional data, where observations are modeled as elements of separable Hilbert spaces. We have derived the limiting distribution of our proposed test statistic and proved the large sample consistency of the practically implemented version of the test based on random permutations. We have also proposed a new local asymptotic normality result for functional data and proved that our test is statistically efficient in the Pitman sense. We have amply demonstrated the superior performance of our test over the state-of-the-art distance-based tests through simulations and real data analysis. Based on our empirical experience, we would highly recommend using $\phi(z) = \sqrt{z}/2$ when the data are nearly orthogonal, i.e., when $\sum_{i=1}^n\sum_{j=1}^m|\langle X_i,Y_j\rangle|/nm$ is small, otherwise we recommend using $\phi(z)= 1-\exp(-z/2)$ or $\phi(z)= \log(1+z)$ for the practical implementation.

The proposed test can be generalized to $k$-sample problems as well, where one needs to use a suitable $k$-sample criterion on the one-dimensional projections corresponding to different distributions $F_1, F_2,\ldots, F_k$. One can also construct a consistent estimate of this measure and develop a $k$-sample test based on it. The large sample behavior of the resulting test can also be investigated using the theory presented in this article. 
\vspace{0.1in}

\textbf{Acknowledgements:} The author would like to thank Anil K. Ghosh and Bhaswar B. Bhattacharya for their helpful comments. The author also thanks Gina-Maria Pomann and Sujit Ghosh for sharing their R codes implementing the FAD test.


\bibliographystyle{apalike} 
\bibliography{refs.bib}

\begin{appendix}
\section*{Appendix}

\begin{lemmaA}
    Let $\{X_n\}$ be a sequence of i.i.d. random variables on the measurable space $(\mathcal{X},\mathcal{A})$ from the distribution $P$. Let $h:\mathcal{X}^k\to \R$ be a measurable function such that $\E_Ph^2(X_1,X_2,\ldots, X_k)$, $\E_Ph^2(X_1,X_1,\ldots, X_1)$ are finite and $\E_P h(x_1, x_2,\ldots, x_{k-1},X_K) = 0$ almost surely. 
    Then the random variable $\int h(x_1,x_2,\ldots, x_k)\prod_{i=1}^k d\hat G_P(x_i)$ converges in distribution to $\int h(x_1,x_2,\ldots, x_k)\prod_{i=1}^kd\mathbb{G}_P(x_i)$ 
    as $n$ goes to infinity, where $\hat{G}_P = \sqrt{n}(P_n-P)$ (the empirical process based on $X_1,X_2,\ldots,X_n$) and $\mathbb{G}_P$ is the $P-$Brownian Bridge process. 
    \label{limit-vstat-1}
\end{lemmaA}

\begin{proof}[\bf Proof]  
 Let $1=f_0,f_1,f_2,\ldots$ be an orthonormal basis of $L_2(\mathcal{X},\mathcal{A},P)$. Since, $h\in L_2(\mathcal{X}^k,\mathcal{A}^k,P^k)$ and is degenerate, we can write $h = \sum_{i_1,i_2,\ldots, i_k}\langle h,f_{i_1}\times \cdots\times f_{i_k}\rangle f_{k_1}\times\ldots\times f_{i_k}$ where $i_1,i_2,\ldots, i_k\geq 1$. Define,  $V_n(f):=\int h(x_1,x_2,\ldots, x_k)\prod_{i=1}^k d\hat G_P(x_i)$. Then for any $g,h\in L_2(\mathcal{X}^k,\mathcal{A}^k,P^k)$,

  \begin{equation*}
      \begin{split}
          \text{cov}(V_n(g),V_n(h)) = & \frac{1}{n^{k-1}} \text{cov}(\Tilde{h}(X_1,X_1,\ldots, X_1),\Tilde{g}(X_1,X_1,\ldots, X_1))\\
          & +k!\frac{ n(n-1)\cdots(n-k+1)}{n^k}\int \Tilde{h}(x_1,x_2,\ldots,x_k)\Tilde{g}(x_1,x_2,\ldots,x_k)\prod_{i=1}^k dP(x_i),
      \end{split}
  \end{equation*}
  $V_n(h)$ is linear in $h$ and $\E\big\{V_n(h)\big\} = \E\big\{\Tilde{h}(X_1,X_1,\ldots,X_1)\big\},$
  where 
  $$\Tilde{h}(u_1,u_2\ldots,u_k) = \int h(x_1,x_2\ldots,x_k)\prod_{i=1}^k d(\delta_{u_i}-P)(x_i).$$
  
  $\Tilde{g}$ is defined analogously. By the degeneracy of $h$, we have $\Tilde{h}=h$. Since $\E_Ph^2(X_1,X_2,\ldots, X_k)$ and $\E_Ph^2(X_1,X_1,\ldots, X_1)$ are finite, the partial sum of the series $h_l = \sum_{i_1=1}^l\sum_{i_2=1}^l\cdots\sum_{i_k=1}^l\langle h,f_{i_1}\times\cdots\times f_{i_k}\rangle f_{i_1}\times\cdots\times f_{i_k}$ converges to $h$ in $L_2(P^k)$. Hence, we get that the series $\sum_{i_1,i_2,\ldots, i_k \geq 1}\langle h,f_{i_1}\times\cdots\times f_{i_k}\rangle V_n(f_{i_1}\times\cdots\times f_{i_k})$ converges in $L_2(P^n)$ (since $\text{Var}\big(V_n(h-h_l)\big)$ and $|\E V_n(h-h_l)|$ both converges to zero as $l$ diverges to infinity). 
  
   Also note that the finite-dimensional distributions of $\{V_n(f_{k_1}\times f_{k_2})\mid (k_1,k_2)\in \mathbb{N}^2\}$ converges to the corresponding finite-dimensional distributions of $\{\int f_{k_1}(u)f_{k_2}(v)d\mathbb{G}_P(u)d\mathbb{G}_P(v)\mid (k_1,k_2)\in \mathbb{N}^2\}$. Then by continuous mapping theorem, $\sum_{i_1,\ldots,i_k\geq 1}\langle h_l,f_{i_1}\times\cdots\times f_{i_k}\rangle V_n(f_{i_1}\times\cdots\times f_{i_k})$ converges in distribution to $\sum_{i_1,\ldots,i_k\geq 1}\langle h_l,f_{i_1}\times\cdots\times f_{i_k}\rangle \prod_{j=1}^k\int f_{i_j}(u)d\mathbb{G}_P(u) =  \int h_l(x_1,x_2,\ldots,x_k)\prod_{i=1}^kd\mathbb{G}_P(x_i)$ as $n$ grows to infinity. 
  
  Now note that, $\sup_{i\in \mathbb{N}}\E \big(\int f_i(u) d\mathbb{G}_P(u)\big)^k <\infty$. Let $[l] := \{1,2,\ldots, l\}$, then for any $l_1<l_2$ we have,
  \begin{equation*}
      \begin{split}
          & \E\big|\int h_{l_2}(x_1,x_2,\ldots,x_k)\prod_{i=1}^kd\mathbb{G}_P(x_i)-\int h_{l_2}(x_1,x_2,\ldots,x_k)\prod_{i=1}^kd\mathbb{G}_P(x_i)\big|\\ 
       = ~& \E\big|\sum_{(i_1,\ldots, i_k)\in [l_2]^k\setminus [l_1]^k} \langle h, f_{i_1}\times\cdots\times f_{i_k}\rangle \prod_{j=1}^k\int f_{i_j}(u) d\mathbb{G}_P(u)\big|\\
    \leq ~& \sum_{(i_1,\ldots, i_k)\in [l_2]^k\setminus [l_1]^k} \big|\langle h, f_{i_1}\times\cdots\times f_{i_k}\rangle \big| \E\big|\prod_{j=1}^k\int f_{i_j}(u) d\mathbb{G}_P(u)\big|\\
    \leq ~& \sum_{(i_1,\ldots, i_k)\in [l_2]^k\setminus [l_1]^k} \big|\langle h, f_{i_1}\times\cdots\times f_{i_k}\rangle \big| \Big\{\sup_{i\in\mathbb{N}}\E\big(\int f_i(u) d\mathbb{G}_P(u)\big)^{k}\Big\}^{1/k}\\
    \leq ~& C \sum_{(i_1,\ldots, i_k)\in [l_2]^k\setminus [l_1]^k} \big|\langle h, f_{i_1}\times\cdots\times f_{i_k}\rangle \big| \leq \big(\sum_{(i_1,\ldots, i_k)\in [l_2]^k\setminus [l_1]^k} \langle h, f_{i_1}\times\cdots\times f_{i_k}\rangle^2\big)^{1/2}
      \end{split}
  \end{equation*}
  Since $h_l$ converges to $h$ in $L_2(P^k)$, it is also a Cauchy sequence. Hence, for any $\epsilon>0$, we can get an $N\in\mathbb{N}$ such that, $\sum_{(i_1,\ldots,i_k)\in [l_2]^k\setminus [l_1]^k} \langle h, f_{i_1}\times\cdots\times f_{i_k}\rangle^2<\epsilon^2$ for any $l_1,l_2\geq N$ and then we also have,
  $$\E\big|\int h_{l_2}(x_1,x_2,\ldots,x_k)\prod_{i=1}^kd\mathbb{G}_P(x_i)-\int h_{l_2}(x_1,x_2,\ldots,x_k)\prod_{i=1}^kd\mathbb{G}_P(x_i)\big|<\epsilon.$$
  
  Hence, the sequence $\int h_{l}(x_1,x_2,\ldots,x_k)\prod_{i=1}^kd\mathbb{G}_P(x_i)$ is Cauchy in $L_1$. It is easy to see that the limit of this sequence of random variables is $\int h(x_1,x_2,\ldots,x_k)\prod_{i=1}^kd\mathbb{G}_P(x_i)$.

  Let the characteristic function of $\int h_{l}(x_1,x_2,\ldots,x_k)\prod_{i=1}^kd\hat{G}(x_i)$ and $\int h(x_1,x_2,\ldots,x_k)\prod_{i=1}^kd\hat{G}(x_i)$ be denoted by $\phi_{nl}(t)$ and $\phi_n(t)$, respectively. Also, let $\phi_l(t)$ and $\phi(t)$ be the characteristic function of $\int h_{l}(x_1,x_2,\ldots,x_k)\prod_{i=1}^kd\mathbb{G}_P(x_i)$ and $\int h(x_1,x_2,\ldots,x_k)\prod_{i=1}^kd\mathbb{G}_P(x_i)$, respectively. Then by the previous arguments we can say that for any $\epsilon>0$ and any $t\in\R$, we can find $K\in\mathbb{N}$ such that,
  $|\phi_{l}(t)-\phi(t)|<\epsilon$ for all $l\geq K$. Therefore, it is easy to see that for all $t\in\R$, and for all $l\geq K$, 
  \begin{equation*}
      \begin{split}
          \lim_{n\to\infty}|\phi_{n}(t)-\phi(t)| & \leq \lim_{n\to\infty}|\phi_{nl}(t)-\phi_n(t)|+ \lim_{n\to\infty}|\phi_{nK}(t)-\phi_{K}(t)|+|\phi_l(t)-\phi(t)|\\
          & = \lim_{n\to\infty}|\phi_{nl}(t)-\phi_n(t)|+\epsilon\\
          & \leq |t| \lim_{n\to\infty}\Big\{\E\big\{|\int h_{l}(x_1,x_2,\ldots,x_k)\prod_{i=1}^kd\hat{G}(x_i)-\int h(x_1,x_2,\ldots,x_k)\prod_{i=1}^kd\hat{G}(x_i)|^2\big\}\Big\}^{1/2}+\epsilon\\
          & = |t| \left\{2\E\big\{(h-h_l)^2(X_1,X_2,\ldots, X_k)\big\}+\big\{\E\{(h-h_l)(X_1,X_1,\ldots, X_1)\}\big\}^2\right\}^{1/2}+\epsilon.
      \end{split}
  \end{equation*}
  In the second last inequality we have used the fact that $\E\{|e^{itX}-1|\}\leq |t|\E|X|\leq |t|\{\E|X|^2\}^{1/2}$.  Now as $l$ grows to infinity we have $\lim\limits_{n\to\infty}|\phi_{n}(t)-\phi(t)|<\epsilon$ and as $\epsilon$ is arbitrary we can say $\lim\limits_{n\to\infty}|\phi_{n}(t)-\phi(t)|=0$. Hence, $\int h_{l}(x_1,x_2,\ldots,x_k)\prod_{i=1}^kd\hat{G}(x_i)$ converges in distribution to $\int h_{l}(x_1,x_2,\ldots,x_k)\prod_{i=1}^kdd\mathbb{G}_P(x_i)$.  
\end{proof}

\begin{lemmaA}
    Let $\{X_n\}$ and $\{Y_m\}$ be two independent sequences of i.i.d. random variables on the measurable space $(\mathcal{X},\mathcal{A})$ from distributions $P$ and $Q$ respectively. Let $h:\mathcal{X}^2\to \R$ be a measurable function such that $\E_Ph^2(X_1,Y_1)<\infty$, which is not necessarily symmetric. Assume $\E_P h(x_1, Y_1) = 0$ and $\E_P h(X_1, y_1) = 0$ almost surely. 
    Then,
    $\sqrt{nm}\int h(x_1,y_1)d\hat{G}_P(x_1)d\hat{G}_Q(y_1)$ converges in distribution to $\int h(x_1,y_1)d\mathbb{G}_P(x_1)d\mathbb{G}_Q(y_1),$ as $\min\{n,m\}$ goes to infinity, where $\hat{G}_P = \sqrt{n}(P_n-P)$ and $\hat{G}_Q=\sqrt{m}(Q_m-Q)$ are the empirical processes based on the observations $X_1,X_2,\ldots,X_n$ and $Y_1,Y_2,\ldots,Y_m$ respectively, and $\mathbb{G}_P$ and $\mathbb{G}_Q$ are independent Brownian Bridge processes.
    \label{limit-vstat-2}
\end{lemmaA}

\begin{proof}[\bf Proof]
    Let $1=f_0,f_1,f_2,\ldots$ be an orthonormal basis of $L_2(\mathcal{X},\mathcal{A},P)$ and $1=g_0,g_1,g_2,\ldots$ be an orthonormal basis of $L_2(\mathcal{X},\mathcal{A},Q)$. Define, $\mathcal{F}=\{f_0,f_1,f_2,\ldots\}\cup\{g_0,g_1,g_2,\ldots\}$. Using similar argument as in Lemma A.\ref{limit-vstat-1} we can argue that $\mathcal{F}$ is a Donsker's class of functions.
    
    Let us define, $T_n(f) := \sqrt{nm}\int h(x_1,y_1)d\hat{G}_P(x_1)d\hat{G}_Q(y_1)$. Then for $g,h\in L_2(\mathcal{X}^2,\mathcal{A}^2,P\otimes Q)$, $\text{cov}(T_n(g),T_n(h)) = \int \Tilde{g}(u,v)\Tilde{h}(u,v)dP(u)dQ(v),$ and $\E\big\{T_n(h)\big\} = \E\big\{\Tilde{h}(X_1,Y_1)\big\} = 0,$
    where $\Tilde{h},\Tilde{g}$ are defined as in Lemma A.\ref{limit-vstat-1}. Now write $h$ as a series $\sum_{k_1,k_2}\langle h, f_{k_1}\times g_{k_2}\rangle f_{k_1}\times g_{k_2}$ where $k_1,k_2\geq 1$ (by the degeneracy assumption). Since, $h_l = \sum_{(k_1,k_2)\in [l]^2}\langle h, f_{k_1}\times g_{k_2}\rangle f_{k_1}\times g_{k_2}$ converges to $h$ in $L_2(P\otimes Q)$, we have $T_n(h_l)$ converges to $T_n(h)$ in $L_2(P^n\otimes Q^m)$(as $\text{Var}\big(h-h_l\big)$ converges to zero as $l$ goes to infinity). 
    
    Since $\mathcal{F}$ is a class of Donsker's function, the finite-dimensional distributions of the process $\{T_n(f_{k_1}\times g_{k_2})\mid (k_1,k_2)\in \mathbb{N}^2\}$ converges to the corresponding finite-dimensional distributions of the process $\{\int f_{k_1}(u)g_{k_2}(v)d\mathbb{G}_P(u)d\mathbb{G}_Q(v)\mid (k_1,k_2)\in \mathbb{N}^2\}$. Then $\sum_k\langle h_l,f_{k_1}\times g_{k_2}\rangle T_n(f_{k_1}\times g_{k_2})$ converges in distribution to $\sum_k\langle h_l,f_{k_1}\times g_{k_2}\rangle \int f_{k_1}(u)g_{k_2}(v)d\mathbb{G}_P(u)d\mathbb{G}_Q(v) = \int h_l(x_1,x_2)d\mathbb{G}_P(x_1)d\mathbb{G}_Q(x_2)$ as $\min\{n,m\}$ grows to infinity, by continuous mapping theorem. 
    
    Now, $\E \big|\int f_i(u) d\mathbb{G}_P(u)\big| = \sqrt{2/\pi}\big(\int f_i^2(u)dP(u)\big)^{1/2} = \sqrt{2/\pi}~\forall i$, by orthonormality of the basis elements. Similarly, $\E \big|\int f_i(u) d\mathbb{G}_P(u)\big| = \sqrt{2/\pi}~\forall i$. Then for any $l_1<l_2$ we have, 
    \begin{equation*}
      \begin{split}
          & \E\big|\int h_{l_1}(x_1,x_2)d\mathbb{G}_P(x_1)d\mathbb{G}_Q(x_2)-\int h_{l_2}(x_1,x_2)d\mathbb{G}_P(x_1)d\mathbb{G}_Q(x_2)\big|\\ 
       = ~& \E\big|\sum_{(i,j)\in [l_2]^2\setminus [l_1]^2} \langle h, f_{i}\times g_j\rangle \int f_i(u) d\mathbb{G}_P(u)\int g_j(u) d\mathbb{G}_Q(u)\big|\\
    \leq ~& \sum_{(i,j)\in [l_2]^2\setminus [l_1]^2} \big|\langle h, f_{i}\times g_j\rangle\big| \E\big|\int f_i(u) d\mathbb{G}_P(u)\int g_j(u) d\mathbb{G}_Q(u)\big|\\
    = ~& \sum_{(i,j)\in [l_2]^2\setminus [l_1]^2} \big|\langle h, f_{i}\times g_j\rangle\big|~\E\big|\int f_i(u) d\mathbb{G}_P(u)\big|~\E\big|\int g_j(u) d\mathbb{G}_Q(u)\big| \\
    \leq ~& \sum_{(i,j)\in [l_2]^2\setminus [l_1]^2} \big|\langle h, f_{i}\times f_j\rangle\big|\frac{2}{\pi} \leq \frac{2}{\pi}\Big(\sum_{(i,j)\in [l_2]^2\setminus [l_1]^2} \big|\langle h, f_{i}\times f_j\rangle\big|^2\Big)^{1/2}.
      \end{split}
  \end{equation*}
    Now using similar arguments as in Lemma A.\ref{limit-vstat-1} gives us our result. 
\end{proof}

\begin{rem}
    Lemma A.\ref{limit-vstat-2} also holds for any $h:\mathcal{X}^k\to\R$ that is completely degenerate and the integrations with respect to $\hat G_P$ and $\hat G_Q$ are taken on any of the coordinates. The proof would be similar to the above, so we omit it here.
\end{rem}

\begin{proof}[\bf Proof of Theorem \ref{thm1}]
Let $X\sim F$ and $Y\sim G$ and define $\phi:\H\to {\mathbb C}$ as
$\phi(f) = \E\{e^{i\langle X,f \rangle}\}-\E\{e^{i\langle Y,f \rangle}\}.$
Notice that the function $\phi$ is continuous and $\phi(f)  = 0~ \forall f\in\H$ implies the two distributions $F$ and $G$ are equal. It is easy to see that if $F=G$, $T(F^f,G^f)=0$ for any $f\in\H$. Hence, when $F=G$, $\zeta^\nu(F,G)=0$ for any probability distribution $\nu$ on $\H$ follows trivially. Now  $\zeta^\nu(F,G) = 0$ implies there exists a Borel measurable set $E\in\mathcal{B}(\H)$ such that $\nu(E) = 1$ and $T(F^f,G^f) = 0~ \forall f\in E$. Hence, by the assumption on $T$, we have
$\phi(f) = 0 ~\forall f\in E$. Thus when $supp\{\nu\}$ is contains the surface of the unit sphere, it follows that for any $f\in\H$, $\langle X,f\rangle = \|f\|\langle X,f/\|f\|\rangle$ and $\langle Y,f\rangle = \|f\|\langle Y,f/\|f\|\rangle$ have the same distribution. Hence, we have $\phi(f) = 0~\forall f\in\H$. 
\end{proof}

\begin{proof}[\bf Proof of Theorem \ref{thm2}]
Let us first consider the following claim.

\textbf{Claim:} Suppose that $X$ is a random variable of a Hilbert space $\H$ with distribution $F$. If $supp\{F\}\subset \H_0$ where $\H_0$ is a closed subspace of $\H$, then $X$ and $QX$ have the same distribution
, where $Q:\H \rightarrow \H_0$ is the projection operator onto $\H_0$.

The claim can be proven by showing that the the characteristic function of $X$ and $QX$ are identical. Now take any arbitrary $\theta \in\H$, and note that
\begin{equation*}
    \begin{split}
        \E\{e^{i\langle X,\theta \rangle}\}  =& \int_\H e^{i\langle x,\theta\rangle} \,dF(x)
         = \int_{\H_0} e^{i\langle x,\theta\rangle} \,dF(x)
         = \int_{\H_0} e^{i\langle Qx,Q\theta\rangle+i\langle (I-Q)x,(I-Q)\theta\rangle}\,dF(x)\\
         & = \int_{\H_0} e^{i\langle Qx,Q\theta\rangle+i\langle 0,(I-Q)\theta\rangle}\,dF(x)
         = \int_{\H_0} e^{i\langle Qx,\theta\rangle}\,dF(x) = \E\{e^{i\langle QX,\theta \rangle}\}.
    \end{split}
\end{equation*}

Hence the claim holds. Denote $\H_0 = \overline{span\{supp\{F\}\cup supp\{G\}\}}$, clearly $\H_0$ is a closed subspace of $\H$. Let $Q$ denote the projection operator from $\H$ to $\H_0$. Then 
\begin{equation*}
    \begin{split}
     \phi(f)  &= \E\{e^{i\langle X,f \rangle}\}-\E\{e^{i\langle Y,f \rangle}\}
         = \int_{\H}e^{i\langle x,f \rangle}dF(x) - \int_{\H}e^{i\langle y,f \rangle}dG(y)\\
        & = \int_{\H_0}e^{i \langle x,f \rangle}dF(x) - \int_{\H_0}e^{i \langle y,f \rangle}dG(y)\\
        & = \int_{\H_0}e^{i \langle Qx,Qf \rangle+i \langle (I-Q)x,(I-Q)f \rangle}dF(x) - \int_{\H_0}e^{i\langle Qy,Qf \rangle+i \langle (I-Q)y,(I-Q)f \rangle}dG(y)\\
        & = \int_{\H_0}e^{i \langle Qx,Qf \rangle}dF(x) - \int_{\H_0}e^{i\langle Qy,Qf \rangle}dG(y)\\
        & {=} \int_{\H_0}e^{i \langle x,Qf \rangle}dF(x) - \int_{\H_0}e^{i \langle y,Qf \rangle}dG(y)~~~~~~~~~\mbox{ (Using Claim with $\theta = Qf$) }\\
        & = \phi(Qf)
    \end{split}
\end{equation*}

This shows that it is enough to consider the value of $\phi(.)$ on $\H_0$. Now if $\zeta^{(F+G)/2}(F,G) = 0$, we have $\phi(f) = 0$ for all $f\in supp\{F\}\cup supp\{G\}$. So, $\langle X, f\rangle$ and $\langle Y, f\rangle$ are independent and identically distributed for any fixed $f\in supp\{F\}\cup supp\{G\}$. Now for any $g\in\H_0$ we can write $g = \sum_{i=1}^\infty a_{i} f_{i}$ for $\{f_{i}\}\subset supp\{F\}\cup supp\{G\}$ and $\{a_{i}\}\subset \R$. It follows trivially that $\langle X,g\rangle = \sum_{i=1}^\infty a_{i}\langle X, f_{i}\rangle$ and $\langle Y,g\rangle = \sum_{i=1}^\infty a_{i}\langle Y, f_{i}\rangle$ are also independent and identically distributed. This implies $\phi(g)= 0 ~\forall g\in\H_0$ and by Claim 1 and 2 $\phi(g) = 0~ \forall g\in\H$. So, if $\zeta^{(F+G)/2}(F,G) = 0$, $X$ and $Y$ are identically distributed. 
\end{proof}

\noindent
\begin{proof}[\bf Proof of Proposition \ref{depprop}]
\text{  }
\begin{itemize}
    \item[(a)] It is easy to see that,
    \begin{equation*}
        \begin{split}
            \E\{g(X_1,X_2,X_3;Y_1,Y_2,Y_3)\} & = \frac{1}{2}\int T_\phi(F^f,G^f)dF(f)+\frac{1}{2}\int T_\phi(F^f,G^f)dG(f)\\
            & = \frac{1}{2}\int T_\phi(F^f,G^f)d(F+G)(f)\\
            & = \zeta_\phi(F,G).
        \end{split}
    \end{equation*}
    \item[(b)] The proof of this statement follows from Theorem \ref{thm2}.
    \item[(c)] Note, if $U:\mathcal{H}\to\mathcal{G}$ is a unitary operator ($U$ is a linear bijective map with $\langle Ux,Uy\rangle = \langle x,y\rangle$), then $g(.)$ remains invariant of this operation, i.e., 
    $$g(UX_1,UX_2,UX_3;UY_1,UY_2,UY_3) = g(X_1,X_2,X_3;Y_1,Y_2,Y_3).$$
    
    Thus $\zeta(F,G) = \E\{g(X_1,X_2,X_3;Y_1,Y_2,Y_3)\} = \E\{g(UX_1,UX_2,UX_3;UY_1,UY_2,UY_3)\} = \zeta(F\circ U^{-1}, G\circ U^{-1})$. Thus $\zeta$ is invariant under unitary operations.
    
    \item[(d)] This follows by simply applying the Dominated Convergence Theorem.
\end{itemize}
\end{proof}

\begin{proof}[\bf Proof of Theorem \ref{largesampleres}]

First let us note that $\zeta(F,G) = \E\{g^*(X_1,X_2,X_3;Y_1,Y_2,Y_3)\}$ and define the function, $\varphi:\R^2\to\R$ as $\varphi(t,s) = \zeta(F+t\sqrt{n}(F_n-F),G+s\sqrt{m}(G_m-G))$, where $F_n$ and $G_m$ are the empirical probability distribution based on the observed data $\mathcal{X}$ and $\mathcal{Y}$. Clearly, $\varphi$ is a bivariate polynomial with random coefficients. It is easy to see that the coefficients are tight (by Lemma A.\ref{limit-vstat-1} and A.\ref{limit-vstat-2}). Also note that $\hat\zeta_{n,m} = \varphi(1/\sqrt{n},1/\sqrt{m})$ and $\varphi(0,0) = \zeta(F,G)$. Hence the limiting distribution of $\hat\zeta_{n,m}$ will be determined by the leading non-zero coefficient of $\varphi(1/\sqrt{n},1/\sqrt{m})$. Now, for any $n,m\in\mathbb{N}$,
$$\varphi(1/\sqrt{n},1/\sqrt{m}) = \varphi(0,0)+\frac{1}{\sqrt{n}}\frac{\partial}{\partial t}\varphi(t,s)\Big|_{(0,0)}+\frac{1}{\sqrt{m}}\frac{\partial}{\partial s}\varphi(t,s)\Big|_{(0,0)}+R_{n,m},$$

where $R_{n,m}$ is of stochastic order $O_P\big(\{1/\sqrt{n}+1/\sqrt{m}\}^2\big)$. Note that, 
$$\frac{\partial}{\partial t}\varphi(t,s)\Big|_{(0,0)} = \sqrt{n}\int h_1^*(x)d(F_n-F)(x)\hspace{15pt}\text{ and }\hspace{15pt}\frac{\partial}{\partial s}\varphi(t,s)\Big|_{(0,0)} = \sqrt{m}\int h_2^*(y)d(G_m-G)(y),$$
where $h_1^*(x) = 3~\E\{g^*(x,X_2,X_3;Y_1,Y_2,Y_3)\}$ and $h_2^*(y) = 3~\E\{g^*(X_1,X_2,X_3;y,Y_2,Y_3)\}$, the first order projection of the symmetrized core function $g^*$ upto a constant term. Hence, we have,
$$\hat \zeta_{n,m}-\zeta(F,G) = \frac{1}{\sqrt{n}}\int h_1^*(x)\sqrt{n}d(F_n-F)(x)+\frac{1}{\sqrt{m}}\int h_2^*(y)\sqrt{m}d(G_m-G)(y)+O_{P}\Big(\big\{\frac{1}{\sqrt{n}}+\frac{1}{\sqrt{m}}\big\}^2\Big).$$

Since, $\mathcal{X}$ and $\mathcal{Y}$ are independently generated, the corresponding empirical processes $\sqrt{n}(F_n-F)$ and $\sqrt{m}(G_m-G)$ are also independent of each other. Now if atleast one $h_1^*$ and $h_2^*$ is a non-zero function, applying CLT we have, $\sqrt{nm/n+m}\big\{\hat \zeta_{n,m}-\zeta(F,G)\big\}$ converges in distribution to a normal limiting distribution with mean zero and variance $\lambda \text{Var}\big(h_1^*(X_1)\big)+(1-\lambda)\text{Var}\big(h_2^*(Y_1)\big)$.

If $h_1^*$ and $h_2^*$ are both zero functions, the limiting distribution of $\hat\zeta_{n,m}$ will be determined by coefficients of the higher order terms in the random bivariate polynomial $\varphi(t,s)$. Generally, in such situations, one may have,
$$\zeta_{n,m}-\zeta(F,G) = \frac{1}{2n}\frac{\partial^2}{\partial t^2}\varphi(t,s)\Big|_{(0,0)}+\frac{1}{2m}\frac{\partial^2}{\partial s^2}\varphi(t,s)\Big|_{(0,0)}+\frac{1}{\sqrt{nm}}\frac{\partial^2}{\partial t\partial s}\varphi(t,s)\Big|_{(0,0)} + R_{n,m}.$$

Where $$\frac{\partial^2}{\partial t^2}\varphi(t,s)\Big|_{(0,0)} = n\int h_1^*(u,v)d(F_n-F)(u)d(F_n-F)(v),$$ 
$$\frac{\partial^2}{\partial t\partial s}\varphi(t,s)\Big|_{(0,0)} = \sqrt{nm}\int h_2^*(u,v)d(F_n-F)(u)d(G_m-F)(v),$$
and $$\frac{\partial^2}{\partial s^2}\varphi(t,s)\Big|_{(0,0)} = m\int h_3^*(u,v)d(G_m-F)(u)d(G_m-F)(v),$$
where $h_1^*(u,v), h_2^*(u,v)$ and $h_3^*(u,v)$ are the second order projections of the symmetrized core function $g^*$ upto a constant factor and $R_{n,m} = O_P\big(\{1/\sqrt{n}+1/\sqrt{m}\}^3\big)$. Here the limiting distribution of $nm/(n+m)\big\{\hat \zeta_{n,m}-\zeta(F,G)\big\}$ will depend on the limiting distribution of the random vector $(\frac{\partial^2}{\partial t^2}\varphi(t,s)\Big|_{(0,0)},\frac{\partial^2}{\partial s^2}\varphi(t,s)\Big|_{(0,0)},\frac{\partial^2}{\partial t\partial s}\varphi(t,s)\Big|_{(0,0)})$. Now we prove our result below.

\begin{enumerate}
    \item[(a)] First we will find the functions $h_1^*(x)$ and $h_2^*(y)$ upto an additive constant term. Since adding a constant to the functions does not change the limiting distribution of the empirical stochastic integrals. Note that,
    \begin{equation*}
        \begin{split}
                h_1^*(x) & =3~\E\big\{g^*(x,X_2,X_3;Y_1,Y_2,Y_3)\big\}\\
                & = \Big\{\E\big(g(x,X_2,X_3;Y_1,Y_2,Y_3)+g(X_1,X_2,x;Y_1,Y_2,Y_3)+g(X_1,X_2,x;Y_1,Y_2,Y_3)\big)\Big\},
        \end{split}
    \end{equation*}
    where $g(x_1,x_2,x_3;y_1,y_2,y_3)$ is the function defined in Proposition \ref{depprop} (a). Now,
    \begin{equation*}
    \begin{split}
        & \E\big(g(x,X_2,X_3;Y_1,Y_2,Y_3)\big)\\
        & =  \frac{1}{2}\E\Big\{-\phi\big(|\langle X_2,x\rangle- \langle X_3,x\rangle|^2\big)-\phi\big(|\langle Y_2,x\rangle- \langle Y_3,x\rangle|^2\big)
         + 2~\phi\big(|\langle X_2,x\rangle - \langle Y_1,x\rangle|^2\big)\\
         & \hspace{40pt}+2~\phi\big(|\langle x,Y_1\rangle- \langle Y_2, Y_1\rangle|^2\big)\Big\}+c_1,\\
        & \E\big(g(X_1,x,X_3;Y_1,Y_2,Y_3)\big)\\
        & =\frac{1}{2}\E\Big\{-\phi\big(|\langle x,X_1\rangle- \langle X_3,X_1\rangle|^2\big)+2~\phi\big(|\langle x,X_1\rangle, \langle Y_1,X_1\rangle|^2\big)
         -\phi\big(|\langle x,Y_1\rangle- \langle X_3,Y_1\rangle|^2\big)\Big\}+c_2,\\
        & \text{and}\\ 
        & \E\big(g(x,X_2,X_3;Y_1,Y_2,Y_3)\big) =  \frac{1}{2}\E\Big\{-\phi\big(|\langle X_2,X_1\rangle- \langle x,X_1\rangle|^2\big)-\phi\big(|\langle X_2,Y_1\rangle- \langle x,Y_1\rangle|^2\big)\Big\}+c_3,
    \end{split}
    \end{equation*}
     where $c_1,c_2$ and $c_3$ are constants that depend on he kernel $k(.,.)$ and the distributions $F$ and $G$. Clearly,
     \begin{equation*}
        \begin{split}
            h_1^*(x)~ = & \Big[\E\Big\{-\phi\big(|\langle x,X_1\rangle- \langle X_3,X_1\rangle|^2\big)-\phi\big(|\langle X_2,Y_1\rangle- \langle x,Y_1\rangle|^2\big)\\
            & \hspace{20pt}+\phi\big(|\langle x,X_1\rangle- \langle Y_1,X_1\rangle|^2\big)+\phi\big(|\langle x,Y_1\rangle- \langle Y_2, Y_1\rangle|^2\big)\Big\}\Big]+\frac{1}{2}T_\phi(F^x,G^x)+c.
        \end{split} 
     \end{equation*}
     Under $H_0$, $h_1^*$ is the zero function and under $H_1$, it is non-zero. We can conclude the same for $h_2(y)$ due to structural symmetry of $g$. Hence, $\sqrt{nm/(n+m)}\big(\hat\zeta_{n,m}-\zeta(F,G)\big)$ converges in distribution to $\sqrt{1-\lambda}\int h_1^*(x)d\mathbb{G}_F(x)+\sqrt{\lambda}\int h_2^*(y)d\mathbb{G}_G(y)$ as $\min\{n,m\}$ goes to infinity with $n/(n+m)\to\lambda\in[0,1]$, where $\mathbb{G}_F$ and $\mathbb{G}_G$ are two independent Brownian Bridge processes. The random variable $\sqrt{1-\lambda}\int h_1^*(x)d\mathbb{G}_F(x)+\sqrt{\lambda}\int h_2^*(y)d\mathbb{G}_G(y)$ is a normal random variable with zero mean and variance $(1-\lambda)\text{Var}\big(h_1^*(X)\big)+\lambda)\text{Var}\big(h_2^*(Y)\big)$.
     
    \item[(b)] Here also we will find the functions $h_1^*(u,v),h_2^*(u,v)$ and $h_3^*(u,v)$ upto an additive constant. Since, for any $g$, $\Tilde{(g+c)}(u,v) = g(u,v)+c-\E h(u,V)-c-\E h(U,v)-c+\E h(U,V)+c = \Tilde{g}$ for any constant $c$, as before, here also the additive constant will not effect the limiting distribution of the empirical stochastic integrals. Note that,
    \begin{equation*}
        \begin{split}
            h_1^*(u,v) ~= \Big\{&\E g\big(u,v,X_3;Y_1,Y_2,Y_3\big)+ \E g\big(u,X_2,v;Y_1,Y_2,Y_3\big)+ \E g\big(X_1,u,v;Y_1,Y_2,Y_3\big)\\
            & + \E g\big(v,u,X_3;Y_1,Y_2,Y_3\big)+\E g\big(v,X_2,u;Y_1,Y_2,Y_3\big)+\E g\big(X_1,u,v;Y_1,Y_2,Y_3\big)\Big\}\\
            =: ~~& \Big\{T_1(u,v)+T_2(u,v)\Big\},
        \end{split}
    \end{equation*}
    where $T_1(u,v) = \E g\big(u,v,X_3;Y_1,Y_2,Y_3\big)+ \E g\big(u,X_2,v;Y_1,Y_2,Y_3\big)+ \E g\big(X_1,u,v;Y_1,Y_2,Y_3\big)$ and $T_1$ and $T_2$ are related as in $T_2(u,v) = T_1(v,u)$. Hence, it is enough to find $T_1(u,v)$. Now,
    \begin{equation*}
    \begin{split}
        & \E\big(g(u,v,X_3;Y_1,Y_2,Y_3)\big) \\
        & =~ \frac{1}{2}\E\Big\{-\phi\big(|\langle v,u\rangle- \langle X_3,u\rangle|^2\big)-\phi\big(|\langle Y_2,u\rangle- \langle Y_3,u\rangle|^2\big)+ 2~\phi\big(|\langle v,u\rangle- \langle Y_1,u\rangle|^2\big)\\
        &~~~-\phi\big(|\langle v,Y_1\rangle- \langle X_3, Y_1\rangle|^2\big)-\phi\big(|\langle Y_2,Y_1\rangle- \langle Y_3, Y_1\rangle|^2\big)+2~\phi\big(|\langle u,Y_1\rangle- \langle Y_2, Y_1\rangle|^2\big)\Big\},\\
        & \E\big(g(u,X_2,v;Y_1,Y_2,Y_3)\big)\\
        & =~ \frac{1}{2}\E\Big\{-\phi\big(|\langle X_2,u\rangle- \langle v,u\rangle|^2\big)-\phi\big(|\langle Y_2,u\rangle- \langle Y_3, u\rangle|^2\big)+2~\phi\big(|\langle X_2,u\rangle- \langle Y_1,u\rangle|^2\big)\\
        & ~~~-\phi\big(|\langle X_2,Y_1\rangle- \langle v,Y_1\rangle|^2\big)-\phi\big(|\langle Y_2,Y_1\rangle- \langle Y_3, Y_1\rangle|^2\big)+2~\phi\big(|\langle u,Y_1\rangle- \langle Y_2, Y_1\rangle|^2\big)\Big\},\\
    \end{split}
    \end{equation*}
    and
    \begin{equation*}
        \begin{split}
            & \E\big(g(X_1,u,v;Y_1,Y_2,Y_3)\big)\\
            & =~ \frac{1}{2}\E\Big\{-\phi\big(|\langle u,X_1\rangle- \langle v,X_1\rangle|^2\big)-\phi\big(|\langle Y_2,X_1\rangle- \langle Y_3,X_1\rangle|^2\big)+2~\phi\big(|\langle u,X_1\rangle- \langle Y_1,X_1\rangle|^2\big)\\
            & ~~~-\phi\big(|\langle u,Y_1\rangle- \langle v, Y_1\rangle|^2\big)-\phi\big(|\langle Y_2,Y_1\rangle- \langle Y_3, Y_1\rangle|^2\big)+2~\phi\big(|\langle X_1,Y_1\rangle- \langle Y_2, Y_1\rangle|^2\big)\Big\}.
        \end{split}
    \end{equation*}
    Therefore, under $H_0$,
    \begin{equation*}
        \begin{split}
            T_1(u,v) ~= & \E\Big\{-\phi\big(|\langle v,Y_1\rangle- \langle X_3, Y_1\rangle|^2\big)\Big\}+\E\Big\{\phi\big(|\langle u,X_1\rangle- \langle Y_1, X_1\rangle|^2\big)\Big\}\\
            &+2~\E\Big\{\phi\big(|\langle u,Y_1\rangle- \langle Y_2, Y_1\rangle|^2\big)\Big\}-\E\Big\{\phi\big(|\langle u,Y_1\rangle- \langle v, Y_1\rangle|^2\big)\Big\}+b_1,
        \end{split}
    \end{equation*}
    for some constant $b_1$ depending on the kernel $k(.,.)$ and the distribution $F$ under $H_0$. Similarly we get,
    \begin{equation*}
        \begin{split}
            T_2(u,v) ~= & \E\Big\{-\phi\big(|\langle u,Y_1\rangle- \langle X_3, Y_1\rangle|^2\big)\Big\}+\E\Big\{\phi\big(|\langle v,X_1\rangle- \langle Y_1, X_1\rangle|^2\big)\Big\}\\
            &+2~\E\Big\{\phi\big(|\langle v,Y_1\rangle- \langle Y_2, Y_1\rangle|^2\big)\Big\}-\E\Big\{\phi\big(|\langle v,Y_1\rangle- \langle u, Y_1\rangle|^2\big)\Big\}+b_1,
        \end{split}
    \end{equation*}
    for some constant $b_1$ depending on the kernel $k(.,.)$ and the distribution $F$ under $H_0$. Hence we obtain, under $H_0$,
    \begin{equation*}
        \begin{split}
            h_1^*(u,v) = 2\left\{\E\Big\{-\phi\big(|\langle u,Y_1\rangle- \langle v, Y_1\rangle|^2\big)\Big\}+\E\Big\{\phi\big(|\langle u,Y_1\rangle- \langle Y_2, Y_1\rangle|^2\big)\Big\}+\E\Big\{\phi\big(|\langle v,Y_1\rangle- \langle Y_2, Y_1\rangle|^2\big)\Big\}\right\}+b_1^*.
        \end{split}
    \end{equation*}
    Clearly, this is symmetric and non-zero. By the structural symmetry of $g$ we also have $h_3^*(u,v)= h_1^*(u,v)$. Also note that,
    \begin{equation*}
        \begin{split}
            h_2^*(u,v) ~= \Big\{&\E g\big(u,X_2,X_3;v,Y_2,Y_3\big)+ \E g\big(u,X_2,X_3;Y_1,v,Y_3\big)+ \E g\big(u,X_2,X_3;Y_1,Y_2,v\big)\\
            & + \E g\big(X_1,u,X_3;v,Y_2,Y_3\big)+ \E g\big(X_1,u,X_3;Y_1,v,Y_3\big)+ \E g\big(X_1,u,X_3;Y_1,Y_2,v\big)\\
            & + \E g\big(X_1,X_2,u;v,Y_2,Y_3\big)+ \E g\big(X_1,X_2,u;Y_1,v,Y_3\big)+ \E g\big(X_1,X_2,u;Y_1,Y_2,v\big)\Big\}.
        \end{split}
    \end{equation*}

   Upon simplification, this yields that under $H_0$, 
    $$h_2^*(u,v) = \left\{\E\Big\{-\phi\big(|\langle u,Y_1\rangle- \langle Y_2, Y_1\rangle|^2\big)\Big\}-\E\Big\{\phi\big(|\langle v,Y_1\rangle- \langle Y_2, Y_1\rangle|^2\big)\Big\}+2~\E\Big\{\phi\big(|\langle u,Y_1\rangle- \langle v, Y_1\rangle|^2\big)\Big\}\right\}+b_3^*.$$

    Hence, under $H_0$, $h_1^*(u,v) = h_3^*(u,v) = 2~h(u,v)$ (say), then $h_2^*(u,v) = - 2 ~h(u,v)$ which is a non-zero function. Thus, the limiting distribution of $\hat\zeta_{n,m}^\phi$ is determined by the joint asymptotic distribution of $\big(\int h(u,v)d\hat{\mathbb{G}}_F(u)d\hat{\mathbb{G}}_F(v),\int h(u,v)d\hat{\mathbb{G}}_F(u)d\hat{\mathbb{G}}^\prime_F(v),\int h(u,v)d\hat{\mathbb{G}}^\prime_F(u)d\hat{\mathbb{G}}^\prime_F(v)\big)$ where $\hat{\mathbb{G}}_F = \sqrt{n}(\hat F_n-F)$ and $\hat{\mathbb{G}}^\prime_F = \sqrt{m}(\hat G_m-F)$. For joint convergence we need to look at the in distribution convergence of,
    \begin{equation*}
        \begin{split}
            t_1\int h(u,v)d\hat{\mathbb{G}}_F(u)d\hat{\mathbb{G}}_F(v)+t_2\int h(u,v)d\hat{\mathbb{G}}_F(u)d\hat{\mathbb{G}}^\prime_F(v)+t_3\int h(u,v)d\hat{\mathbb{G}}^\prime_F(u)d\hat{\mathbb{G}}^\prime_F(v)
        \end{split}
    \end{equation*}
    
    for some real numbers $t_1,t_2$ and $t_3$. 
    We first write $h$ as a series with respect to an orthonormal basis $1=f_0,f_1,f_2,\ldots$ of $L_2(F)$ as,
    \begin{equation*}
        \begin{split}
            h(u,v) = \sum_{(k_1,k_2)\in \mathbb{N}^2}\langle h,f_{k_1}\times f_{k_2}\rangle f_{k_1}\times f_{k_2}
        \end{split}
    \end{equation*}

    and truncate the series at $l$ to get,
    \begin{equation*}
        \begin{split}
            h_l(u,v) = \sum_{(k_1,k_2)\in [l]^2}\langle h,f_{k_1}\times f_{k_2}\rangle f_{k_1}\times f_{k_2}.
        \end{split}
    \end{equation*}

    Continuous mapping theorem gives us that the random variable $t_1\int h_l(u,v)d\hat{\mathbb{G}}_F(u)d\hat{\mathbb{G}}_F(v)$ $+t_2\int h_l(u,v)d\hat{\mathbb{G}}_F(u)d\hat{\mathbb{G}}^\prime_F(v)$ $+t_3\int h_l(u,v)d\hat{\mathbb{G}}^\prime_F(u)d\hat{\mathbb{G}}^\prime_F(v)$ converges in distribution to the random variable $t_1\int h_l(u,v)d\mathbb{G}_F(u)d\mathbb{G}_F(v)+ t_2 \int h_l(u,v) d\mathbb{G}_F(u) d\mathbb{G}^\prime_F(v)+t_3\int h_l(u,v)d\mathbb{G}^\prime_F(u)d\mathbb{G}^\prime_F(v)$. Also take any $l_2>l_1>N$ such that $\|h_{l_1}-h_{l_2}\|_{L_2(F^2)}<\epsilon$, then
    \begin{equation*}
        \begin{split}
            \E\Big|&t_1\int h_{l_1}(u,v)d{\mathbb{G}}_F(u)d{\mathbb{G}}_F(v)+t_2\int h_{l_1}(u,v)d{\mathbb{G}}_F(u)d{\mathbb{G}}^\prime_F(v)\\
            & +t_3\int h_{l_1}(u,v)d{\mathbb{G}}^\prime_F(u)d{\mathbb{G}}^\prime_F(v) -t_1\int h_{l_2}(u,v)d{\mathbb{G}}_F(u)d{\mathbb{G}}_F(v)\\
            & -t_2\int h_{l_2}(u,v)d{\mathbb{G}}_F(u)d{\mathbb{G}}^\prime_F(v)-t_3\int h_{l_2}(u,v)d{\mathbb{G}}^\prime_F(u)d{\mathbb{G}}^\prime_F(v)\Big|\\
            \leq ~~~& |t_1|~\E\Big|\int h_{l_1}(u,v)d{\mathbb{G}}_F(u)d{\mathbb{G}}_F(v)-\int h_{l_2}(u,v)d{\mathbb{G}}_F(u)d{\mathbb{G}}_F(v)\Big|\\
            & +|t_2|~\E\Big|\int h_{l_1}(u,v)d{\mathbb{G}}_F(u)d{\mathbb{G}}^\prime_F(v)-\int h_{l_2}(u,v)d{\mathbb{G}}_F(u)d{\mathbb{G}}^\prime_F(v)\Big|\\
            & +|t_3|~\E\Big|\int h_{l_1}(u,v)d{\mathbb{G}}^\prime_F(u)d{\mathbb{G}}^\prime_F(v)-\int h_{l_2}(u,v)d{\mathbb{G}}^\prime_F(u)d{\mathbb{G}}^\prime_F(v)\Big|.\\
        \end{split}
    \end{equation*}
    Now using Lemma A.\ref{limit-vstat-1} and Lemma A.\ref{limit-vstat-2} we can say that the above random variable converges in distribution to 
        \begin{equation*}
        \begin{split}
            t_1\int h(u,v)d&\mathbb{G}_F(u)d\mathbb{G}_F(v)+t_2\int h(u,v)d\mathbb{G}_F(u)d\mathbb{G}'_F(v)+t_3\int h(u,v)d\mathbb{G}'_F(u)d\mathbb{G}'_F(v).
        \end{split}
    \end{equation*}

Hence by continuous mapping theorem under $H_0$, as $\min\{n,m\}$ goes to infinity with $n/(n+m)\to\lambda\in[0,1]$, $nm/(n+m)\zeta_{n,m}^\phi$ converges in distribution to $(1-\lambda)\int h(u,v)d\mathbb{G}_F(u)d\mathbb{G}_F(v)-2\sqrt{\lambda(1-\lambda)}\int h(u,v)d\mathbb{G}_F(u)d\mathbb{G}'_F(v)+\lambda\int h(u,v)d\mathbb{G}'_F(u)d\mathbb{G}'_F(v),$  where $\mathbb{G}_F$ and $\mathbb{G}_F^\prime$ are two independent Brownian Bridge processes. 

Since $h(.,.)$ is symmetric, using the Fredholm theory of integral equations we can also write $h(u,v) = \sum_{i=1}^\infty \lambda_i \varphi_i(u)\varphi_i(v)$, where $\{\lambda_i\}$ and $\{\varphi_i\}$ are the eigenvalues and eigenfunctions of the integral equation $\int h(u,v)\gamma(v)dF(v) = \lambda\gamma(u)$, where the equality holds in the $L_2$ sense. Then the limiting distribution of $nm/(n+m)\hat\zeta_{n,m}^\phi$ can be written as $\sum_{i=1}^\infty \lambda_i\big(\sqrt{1-\lambda}\int \varphi_i(u)d\mathbb{G}_F(u)-\sqrt{\lambda}\int \varphi_i(u)d\mathbb{G}_F^\prime(u)\big)^2$, which is identically distributed as $\sum_{i=1}^\infty \lambda_i Z_i^2$ where $\{Z_i\}$ is a sequence of i.i.d. standard normal random variables. This completes the proof.

\end{enumerate}

\end{proof}

\begin{proof}[\bf Proof of Corollary \ref{consistency}]

This follows easily from Theorem \ref{largesampleres}.
\end{proof}

\begin{proof}[\bf Proof of Theorem \ref{local-limit-per}]
Let $\mathcal{U} = \{U_1=X_1,\ldots, U_n=X_n,U_{n+1}=Y_1,\ldots, U_{N}=Y_m\}$ ($N=n+m$) be the pooled data and $\pi$ be a random permutation of $\{1,2,\ldots,n\}$ independent of the observed data and define the two-sample permutation empirical measure as,
    $$\Tilde{P}_{n,N} = \frac{1}{n}\sum_{i=1}^n\delta_{U_{\pi(i)}}\hspace{20pt}\Tilde{Q}_{m,N} = \frac{1}{m}\sum_{i=n+1}^N\delta_{U_{\pi(i)}}.$$
    Define the pooled empirical measures as
    $$H_N = \frac{1}{N}\sum_{i=1}^N \delta_{U_{i}}.$$

    Then using Theorem 3.7.1 from \cite{wellner2013weak} we have that over a suitable class of functions $\mathcal{F}$, $\sqrt{n}(\Tilde{P}_{n,N}-H_N)$ converges in distribution to $\sqrt{1-\lambda}\mathbb{G}_H$ where $H=\lambda F+(1-\lambda) G$ and $\lim n/N = \lambda$. Since $\sqrt{m}(\Tilde{Q}_{m,N}-H_N) = -\sqrt{n/m}\sqrt{n}(\Tilde{P}_{n,N}-H_N)$, the distributional convergence of $\sqrt{m}(\Tilde{Q}_{m,N}-H_N)$ follows trivially.

    Now applying arguments as in Theorem \ref{largesampleres} (b) we get that for any fixed alternative $nm/(n+m)\hat\zeta_{n,m}^{\phi,\pi}$ converges in distribution to $(1-\lambda)^2\int h(u,v)d\mathbb{G}_H(u)d\mathbb{G}_H(v)+2\lambda(1-\lambda)\int h(u,v)d\mathbb{G}_H(u)d\mathbb{G}_H(v)+\lambda^2\int h(u,v)d\mathbb{G}_H(u)d\mathbb{G}_H(v),$ as $\min\{n,m\}$ diverges to infinity with $n/N$ converging to $\lambda\in [0,1]$, which is same as $\int h(u,v)d\mathbb{G}_H(u)d\mathbb{G}_H(v)$. The limiting distribution has the same distribution as $\sum_{k=1}^\infty\lambda_kZ_k^2$ where $\{\lambda_k\}$ is a square-integrable sequence of real numbers and $\{Z_k\}$ is an i.i.d. sequence of standard normal random variables.

\end{proof}

\begin{proof}[\bf Proof of Proposition \ref{MCpval}]

For proving this let us first define,
$$F(t) = \frac{1}{N!}\left\{\sum_{\pi\in \mathcal{S}_N}\mathbbm{1}[\hat\zeta_{\pi_i}\leq t]\right\}\hspace{20pt}and\hspace{20pt}F_B(t) = \frac{1}{B}\left\{\sum_{i=1}^B\mathbbm{1}[\hat\zeta_{\pi_i}\leq t]\right\}.$$

$F$ and $F_B$ are distribution functions conditioned on the observed pooled data $\mathcal{U}$. Then,
\begin{equation*}
    \begin{split}
        |p_{n,m}-p_{n,m,B}| & = |\frac{1}{N!}\left\{\sum_{\pi\in \mathcal{S}_N}\mathbbm{1}[\hat\zeta_{\pi_i}\geq \hat\zeta_{n,m}]\right\}-\frac{1}{B+1}\left\{\sum_{i=1}^B\mathbbm{1}[\hat\zeta_{\pi_i}\geq \hat\zeta_{n,m}]+1\right\}|\\
        & = |\frac{1}{N!}\left\{\sum_{\pi\in \mathcal{S}_N}\mathbbm{1}[\hat\zeta_{\pi_i}< \hat\zeta_{n,m}]\right\}-\frac{1}{B+1}\left\{\sum_{i=1}^B\mathbbm{1}[\hat\zeta_{\pi_i}< \hat\zeta_{n,m}]\right\}|\\
        &  = |F(\hat\zeta_{n,m}-)-\frac{B}{B+1}F_B(\hat\zeta_{n,m}-)|\\
        & \leq |F(\hat\zeta_{n,m}-)-F_B(\hat\zeta_{n,m}-)| + |\frac{F_B(\hat\zeta_{n,m}-)}{B+1}| \leq \sup_{t\in\R}|F(t)-F_B(t)|+\frac{1}{B+1}
    \end{split}
\end{equation*}

Conditioned on the pooled data $\mathcal{U}$, the Dvoretzky-Keifer-Wolfwitz inequality (\cite{massart1990tight}) gives us,
$\P\{\sup_{t\in\R}|F(t)-F_B(t)|>\epsilon\}\leq 2e^{-2B\epsilon^2}.$ Hence, conditioned on the pooled data $\mathcal{U}$, as $B$ grows to infinity the randomized p-value $p_{n,m,B}$ converges almost surely to $p_{n,m}$.
\end{proof}

\begin{proof}[\bf Proof of Theorem \ref{lanfd}]
    Under Assumption (A1) it is easy to see that $F^{(N)} = (1-\alpha/\sqrt{N})F+\alpha/\sqrt{N}L$ has the Radon-Nikodym derivative as $\Big(1+\frac{\alpha}{\sqrt{N}}\big(\ell(z)-1\big)\Big)$ with respect to $F$. Hence if $Z_1,Z_2,\ldots,Z_N$ are independent and identically generated observations, then the log-likelihood ratio is given by,
    $$L_N = \log\Big\{\prod_{i=1}^N\frac{dF^{(N)}}{dF}(Z_i)\Big\} = \sum_{i=1}^N \log\Big\{\frac{dF^{(N)}}{dF}(Z_i)\Big\} = \sum_{i=1}^N \log\Big(1+\frac{\alpha}{\sqrt{N}}\big(\ell(Z_i)-1\big)\Big).$$

    Using the fact,
    $$\log(1+y) = y - \frac{y^2}{2}+\frac{1}{2}y^2\beta(y)~~~~\mbox{where}~~\lim_{y\to0}\beta(y)=0,$$

    gives us 
    $$L_N = \sum_{i=1}^N \frac{\alpha}{\sqrt{N}}\big(\ell(Z_i)-1\big)-\sum_{i=1}^N \frac{\alpha^2}{2N}\big(\ell(Z_i)-1\big)^2+\sum_{i=1}^N \frac{\alpha^2}{2N}\big(\ell(Z_i)-1\big)^2\beta\Big(\frac{\alpha}{\sqrt{N}}\big(\ell(Z_i)-1\big)\Big).$$

    Under Assumption (A1),
    $$\sum_{i=1}^N \frac{\alpha^2}{N}\big(\ell(Z_i)-1\big)^2\stackrel{a.s.}{\to} \alpha^2 \E\Big(\big(\ell(Z_1)-1\big)^2\Big)$$
    as $N$ grows to infinity under $F$. Hence we only need to show that 
    $$\sum_{i=1}^N \frac{\alpha^2}{N}\big(\ell(Z_i)-1\big)^2\beta\Big(\frac{\alpha}{\sqrt{N}}\big(\ell(Z_i)-1\big)\Big)$$
    converges to zero in probability under $F$. Notice that,
    $$\sum_{i=1}^N \frac{\alpha^2}{N}\big(\ell(Z_i)-1\big)^2\beta\Big(\frac{\alpha}{\sqrt{N}}\big(\ell(Z_i)-1\big)\Big) \leq \max_{1\leq i\leq N}|\beta\Big(\frac{\alpha}{\sqrt{N}}\big(\ell(Z_i)-1\big)\Big)|\sum_{i=1}^N \frac{\alpha^2}{N}\big(\ell(Z_i)-1\big)^2.$$
    Due to Assumption (A1) it suffices to show that $\max_{1\leq i\leq N}|\beta\Big(\frac{\alpha}{\sqrt{N}}\big(\ell(Z_i)-1\big)\Big)|$ converges to zero in probability, which follows if $\max_{1\leq i\leq N}|\frac{\alpha}{\sqrt{N}}\big(\ell(Z_i)-1\big)|$ converges to zero in probability (as $\lim_{y\to0} \beta(y) = 0$). But $\max_{1\leq i\leq N}|\ell(Z_i)-1|$ is a tight random variable, thus the convergence holds trivially. Hence we have,
    $$\left|\log\Big\{\prod_{i=1}^N\frac{dF^{(N)}}{dF}(Z_i)\Big\}-\frac{\alpha}{\sqrt{N}}\sum_{i=1}^N \Big(\ell(Z_i)-1\Big)+\frac{\alpha^2}{2}\E\Big\{\ell(Z_1)-1\Big\}^2\right|\to0,$$
    in probability under $F$ as $N$ goes to infinity.
\end{proof}

\begin{lemmaA}
Under Assumption (A1) and $\delta_N = \alpha/\sqrt{N}$, the process $\{\sqrt{n}\int f(u)\big(\hat F_n-F\big)(u)\mid f\in\mathcal{F}\}$ converges in distribution to the process $\{\int \Tilde{f} d\big(\mathbb{G}_F+\alpha(L-F)\big)\mid \Tilde{f}(u)=f(u)-\E_Ff(Z), f\in \mathcal{F}\}$ in $\ell^\infty(\mathcal{F})$ where $\mathcal{F}$ is a Donsker class of measurable functions with $\sup_{f\in \mathcal{F}}|Pf|<\infty$ and $\mathbb{G}_F$ is the $F$-Brownian Bridge process on $\ell^\infty(\mathcal{F})$.
\label{emp-local}
\end{lemmaA}

\begin{proof}[\bf Proof of Lemma A.\ref{emp-local}]
    Here we only show the finite-dimensional distribution convergence of the empirical process. The tightness of the process under $F^{(N)}$ can be concluded using Theorem 3.10.12 from \cite{wellner2013weak}. Let $Z_1,Z_2,\ldots,Z_N$ be a sequence of independent and identically distributed random variables. Let $F^{(N)} = (1-\alpha/\sqrt{N})F+\alpha/\sqrt{N}L$ where $L$ satisfies Assumption (A1) and the empirical probability distribution be denoted by $\hat F_N$. Take any $f\in \mathcal{F}$ with $\E_F f^2(Z)<\infty$ and note that the joint limiting distribution of $\int f(u)\sqrt{N}d(\hat F_N-F)(u)$ and $\log \big\{\prod_{i=1}^N dF^{(N)}/dF(Z_i)\big\}$ is same as the joint limiting distribution of $\int f(u)\sqrt{N}(\hat F_N-F)(u)$ and $\frac{\alpha}{\sqrt{N}}\sum_{i=1}^N \big(\ell(Z_i)-1\big)-\frac{\alpha^2}{2}\E\big\{\ell(Z_1)-1\big\}^2$ (by Theorem \ref{lanfd} and Slutsky's theorem), which is a bivariate normal distribution with mean and variance-covariance matrix as follows,
    $$\mu = \begin{pmatrix}
    0\\
    -\frac{\alpha^2}{2}\E\big\{\ell(Z_1)-1\big\}^2
    \end{pmatrix}\hspace{0.5in}
    \Sigma = \begin{pmatrix}
    \int \Tilde{f}^2(u)dF(u) & \tau \\
    \tau & \alpha^2\E\big\{\ell(Z_1)-1\big\}^2
    \end{pmatrix},$$
    where $\tau = \alpha\big(\int \Tilde{f}(u)dL(u)\big)$ and $\Tilde{f}(u) = f(u)-\E_Ff(Z)$. According to Le Cam's third lemma, this implies that under $F^{(N)}$, $\int f(u)\sqrt{N}d(\hat F_N-F)(u)$ converges in distribution to a normal distribution with mean $\tau = \alpha\big(\int \Tilde{f}(u)dL(u)\big)$ and variance $\int \Tilde{f}^2(u)dF(u)$. Using Cramer-Wold device one can further show that, under $F^{(N)}$ the finite dimensional distributions of $\big\{\int f(u)\sqrt{N}(\hat F_N-F)(u)\mid f\in\mathcal{F}\big\}$ converges to the finite dimensional distributions of the process $\big\{\int f(u)d\big(\mathbb{G}_F+\alpha(L-F)\big)(u)\mid f\in \mathcal{F}\big\}$. Hence, under the contiguous alternative $F^{(N)}$, the empirical process converges in distribution to the process $\big\{\int f(u)d\big(\mathbb{G}_F+\alpha(L-F)\big)(u)\mid f\in \mathcal{F}\big\}$ where $\mathcal{F}$ is a $F$- Donsker class of functions.  
\end{proof}

\begin{proof}[\bf Proof of Theorem \ref{local-limit}]
    Applying Theorem \ref{emp-local} we get that under $(F^{(n)},G^{(m)})$ the empirical processes $\sqrt{n}(\hat F_n-F)$ and $\sqrt{m}(\hat G_m-F)$ converges in distribution to the processes $\mathbb{G}_F$ and $\mathbb{B}'_F = \mathbb{G}'_F+\delta(L-F)$ respectively, where $\mathbb{G}_F$ and $\mathbb{G}_F'$ are independent Brownian Bridge processes. Now note that if $1=f_0,f_1,f_2,\ldots$ are orthonormal basis of $L_2(F)$ then, it is easy to see $\int f_i(u)d\mathbb{G}_F(u)+\delta\int f_i(u) dL(u)$ has a normal distribution with mean $\delta\int f_i(u) dL(u)$ and variance unity. Also,
    $$|\int f_i(u) d L(u)|=|\int f_i(u)\ell(u) dF(u)|\leq \big(\int \ell^2(u)dF(u)\big)^{1/2}.$$
    Therefore, $\max_i \E\Big|\int f_i(u)d\mathbb{B}'_F(u)\Big| \leq \max_i\E\Big|\int f_i(u)d\mathbb{G}'_F(u)\Big|+|\delta|\big(\int \ell^2(u)dF(u)\big)^{1/2}<\infty.$ Using arguments similar to Theorem \ref{largesampleres} (b), gives us that under $(F^{(n)},G^{(m)})$ $nm/(n+m)\hat\zeta_{n,m}$ converges in distribution to $\sum_{i=1}^\infty \lambda_i\big(\sqrt{1-\lambda}\int \varphi_i(u)d\mathbb{G}_F(u)-\sqrt{\lambda}\int \varphi_i(u)d\mathbb{B}_F^\prime(u)\big)^2$ where $\lim n/(n+m) = \lambda$. Then it is easy to show that the limiting distribution is identically distributed with $\sum_{i=1}^\infty\lambda_i\big(Z_i-\sqrt{\lambda}\delta\int \varphi_i(u)d(L-F)(u)\big)^2$, where $\{\lambda_i\},\{Z_i\}$ and $\{\varphi_i\}$ are as in the proof of Theorem \ref{largesampleres} (b). 
\end{proof}

\begin{proof}[\bf Proof of Theorem \ref{local-limit-per-2}]

    Under the contiguous alternative $(F^{(n)}, G^{(m)})$ (as in Theorem \ref{local-limit}), for any $f$ with finite $\int f^2(u)dF(u)$ we need to find the joint limiting distribution of $(\frac{1}{n}\sum_{i=1}^n f(U_{\pi(i)})-\frac{1}{N}\sum_{i=1}^N f(U_{i}))$ and $\frac{\alpha}{N}\sum_{i=1}^N \big(\ell(U_i)-1\big)-\frac{\alpha^2}{2}\E\big\{\ell(U_1)-1\big\}^2$ assuming $U_i\stackrel{i.i.d.}{\sim} F$. Note that,
    \begin{equation*}
        \begin{split}
    & \begin{pmatrix}
    \big(\frac{1}{n}\sum_{i=1}^n f(U_{\pi(i)})-\frac{1}{N}\sum_{i=1}^N f(U_{i})\big)\\
    \frac{\alpha}{N}\sum_{i=1}^N \big(\ell(U_i)-1\big)
    \end{pmatrix} 
    \stackrel{d}{=} 
    \begin{pmatrix}
    \big(\frac{1}{n}\sum_{i=1}^n f(U_{i})-\frac{1}{N}\sum_{i=1}^N f(U_{i})\big)\\
    \frac{\alpha}{N}\sum_{i=1}^N \big(\ell(U_i)-1\big)
    \end{pmatrix}\\
    & = \begin{pmatrix}
        1 & -1 & 0 \\
        0 & 0 & 1
    \end{pmatrix}
    \begin{pmatrix}
        1 & 0 & 0 \\
        \frac{n}{N} & \frac{m}{N} & 0\\
        0 & 0 & 1
    \end{pmatrix}\begin{pmatrix}
        \frac{1}{n}\sum_{i=1}^n \big(f(U_{i}) -\E f(U_1) \big)\\
        \frac{1}{m}\sum_{i=n+1}^N \big(f(U_{i}) -\E f(U_1) \big)\\
        \frac{\alpha}{N}\sum_{i=1}^N \big(\ell(U_i)-1\big)
    \end{pmatrix}\\
    & = \begin{pmatrix}
        1 & -1 & 0 \\
        0 & 0 & 1
    \end{pmatrix}\begin{pmatrix}
        1 & 0 & 0 & 0\\
        \frac{n}{N} & \frac{m}{N} & 0 & 0\\
        0 & 0 & 1 & 1\\
    \end{pmatrix}\begin{pmatrix}
        \frac{1}{n}\sum_{i=1}^n \big(f(U_{i}) -\E f(U_1) \big)\\
        \frac{1}{m}\sum_{i=n+1}^N \big(f(U_{i}) -\E f(U_1) \big)\\
        \frac{\alpha}{N}\sum_{i=1}^n \big(\ell(U_i)-1\big)\\
        \frac{\alpha}{N}\sum_{i=n+1}^N \big(\ell(U_i)-1\big)
    \end{pmatrix}\\
    & \stackrel{d}{=} \begin{pmatrix}
        1 & -1 & 0 \\
        0 & 0 & 1
    \end{pmatrix}\begin{pmatrix}
        1 & 0 & 0 & 0\\
        \frac{n}{N} & \frac{m}{N} & 0 & 0\\
        0 & 0 & 1 & 1\\
    \end{pmatrix}\begin{pmatrix}
        \frac{1}{n}\sum_{i=1}^n \big(f(U_{i}) -\E f(U_1) \big)\\
        \frac{1}{m}\sum_{i=1}^m \big(f(U_{i}^\prime) -\E f(U_1) \big)\\
        \frac{\alpha}{N}\sum_{i=1}^n \big(\ell(U_i)-1\big)\\
        \frac{\alpha}{N}\sum_{i=1}^m \big(\ell(U_i^\prime)-1\big)
    \end{pmatrix}\\
    & =: \begin{pmatrix}
        1 & -1 & 0 \\
        0 & 0 & 1
    \end{pmatrix} \begin{pmatrix}
        1 & 0 & 0 & 0\\
        \frac{n}{N} & \frac{m}{N} & 0 & 0\\
        0 & 0 & 1 & 1\\
    \end{pmatrix}\begin{pmatrix}
        W_{n1}\\
        W_{n2}\\
        W_{n3}\\
        W_{n4}\\
    \end{pmatrix}\\
        \end{split}
    \end{equation*}
    where $U_1^\prime,U_2^\prime,\ldots,U_m^\prime$ is an i.i.d sample from $F$ independent of $U_1,U_2,\ldots,U_n$. Applying multivariate CLT and continuous mapping theorem we get that $\sqrt{n}\Big(\big(\frac{1}{n}\sum_{i=1}^n f(U_{\pi(i)})-\frac{1}{N}\sum_{i=1}^N f(U_{i})\big),$ $ \frac{\alpha}{N}\sum_{i=1}^N \big(\ell(U_i)-1\big)\Big)$ converges in distribution to $\big((1-\lambda)W_1-\sqrt{\lambda(1-\lambda)} W_2, \lambda W_3 + \sqrt{\lambda(1-\lambda)}W_4)\big)$, where 
    $(W_1,W_2,W_3,W_4)$ has a multivariate normal distribution with zero mean and variance-covariance matrix as follows,
    $$\begin{pmatrix}
        \int \Tilde{f}^2(u) dF(u) & 0 & \alpha\int \Tilde{f}(u)dL(u) & 0\\
        0 & \int \Tilde{f}^2(u) dF(u) & 0 & \alpha\int \Tilde{f}(u)dL(u)\\
        \alpha\int \Tilde{f}(u)dL(u) & 0 & \alpha^2 \int (\ell(u)-1)^2dF(u) & 0\\
        0 & \alpha\int \Tilde{f}(u)dL(u) & 0 & \alpha^2 \int (\ell(u)-1)^2dF(u)\\
    \end{pmatrix},$$
    where $\Tilde{f}(u)=f(u)-\E_F f(U)$. Hence, the random vector $\sqrt{n}\Big(\big(\frac{1}{n}\sum_{i=1}^n f(U_{\pi(i)})-\frac{1}{N}\sum_{i=1}^N f(U_{i})\big),$ $ \frac{\alpha}{N}\sum_{i=1}^N \big(\ell(U_i)-1\big)\Big)$ converges in distribution to a normal random variable with mean zero and variance-covariance matrix 
    $$\begin{pmatrix}
        (1-\lambda)\int \Tilde{f}^2(u)dF(u) & 0 \\
        0 & \lambda \alpha^2 \int \big(\ell(u)-1\big)^2 dF(u)\\
    \end{pmatrix}.$$ 
    
    where $\Tilde{f}(u) = f(u)-\E_Hf(U)$. Hence, using Le Cam's third lemma and Cramer-Wold device we can say that the finite-dimensional distribution of the process $\sqrt{n}(\Tilde{P}_{n,N}-H_N)$ converges in distribution to the finite-dimensional distribution of the process $\sqrt{1-\lambda}\mathbb{G}_F$ where $\Tilde{P}_{n,N}$ is the permutation empirical measure as defined in the proof of Theorem \ref{local-limit-per}. The tightness of the process follows from Theorem 3.10.12 of \cite{wellner2013weak}. Now applying arguments as in Theorem \ref{largesampleres} (b) we get that under $(F^{(n)}, G^{(m)})$, $nm/(n+m)\hat\zeta_{n,m}^{\phi,\pi}$ converges in distribution to $\sum_{k=1}^\infty\lambda_kZ_k^2$ for some square integrable sequence $\{\lambda_i\}$ and a sequence of independent standard normal random variable $\{Z_i\}$. This completes the proof.
\end{proof}
\end{appendix}

\end{document}